\def \VersionArXiv {}

\ifdefined \VersionArXiv
	\newcommand{\arXivVersion}[1]{#1}
	\newcommand{\CONCURVersion}[1]{}
\else
	\newcommand{\arXivVersion}[1]{}
	\newcommand{\CONCURVersion}[1]{#1}
\fi

\documentclass[a4paper,UKenglish,cleveref, autoref, thm-restate]{lipics-v2021}
\arXivVersion{
	\hideLIPIcs  %
}

\usepackage[utf8]{inputenc}
\usepackage{csquotes}
  \usepackage{breakurl}           %

\usepackage{thm-restate}
\usepackage{xspace}
\usepackage{adjustbox}
\usepackage{amsmath}
\usepackage{amssymb} %

\usepackage{subcaption} %

\usepackage{pdflscape} %

\usepackage{paralist} %

\newenvironment{ienumerate}
	{\ifdefined\VersionArXiv\begin{enumerate}\else\begin{inparaenum}[\itshape i\upshape)]\fi}
	{\ifdefined\VersionArXiv\end{enumerate}\else\end{inparaenum}\fi}

\newenvironment{oneenumerate}
	{\ifdefined\VersionArXiv\begin{enumerate}\else\begin{inparaenum}[1)]\fi}
	{\ifdefined\VersionArXiv\end{enumerate}\else\end{inparaenum}\fi}

\ifdefined\VersionWithComments
\usepackage{everypage}

\AddEverypageHook{%
	\begin{tikzpicture}[remember picture, overlay]
		\node[
		align=center,
		anchor=north east,
		xshift=-1cm,
		yshift=-1cm,
		text=red
		] at (current page.north east) {\scalebox{2}{\bfseries\LARGE Draft!}\\{\tiny\today{}}};
	\end{tikzpicture}
}

\fi
\PassOptionsToPackage{table}{xcolor}
\newcommand{\cellHeader}[1]{\cellcolor{blueColorBlind!40}\textbf{#1}}
\newcommand{\rowHeader}{\rowcolor{blueColorBlind!40}}

\newcommand{\cellDecidable}{\cellcolor{greenColorBlind!50}} %
\newcommand{\cellUndecidable}{\cellcolor{redColorBlind!50}} %
\newcommand{\cellDecidableExisting}{\cellcolor{greenColorBlind!25}} %
\newcommand{\cellUndecidableExisting}{\cellcolor{redColorBlind!25}} %

\usepackage{tikz}
\usetikzlibrary{arrows,automata,positioning,shapes} %
\usetikzlibrary{decorations} %
\usetikzlibrary{decorations.pathreplacing} %
\usetikzlibrary{calc} %

\tikzstyle{pta}            = [auto, ->, >=stealth',initial text=]
\tikzstyle{every node}=[initial text=]
\tikzstyle{RA-location}=[rectangle, rounded corners, minimum size=12pt, draw=black, fill=blueColorBlind!15, inner sep=2pt]
\tikzstyle{location}=[rounded rectangle, minimum size=12pt, draw=black, fill=blueColorBlind!15, inner sep=2pt]
\tikzstyle{final}=[double, fill=greenColorBlind!40]

\tikzstyle{private}=[fill=redColorBlind!30,thick]

\tikzset{
  location2gen/.style={
    draw=black,
    rounded corners,
	align=center,
    inner sep=2pt,
    font=\small
  }
}

\tikzset{
  location2/.style={
	location2gen,
    rectangle split,
    rectangle split parts=2,
    rectangle split horizontal=false,
	rectangle split part fill={blueColorBlind!5, blueColorBlind!15},
  }
}

\newcommand{\location}[3][]{%
  \node[location2,#1] {%
    \scriptsize $\ensuremath{$#2$}$
    \nodepart{second}
    #3
  };
}

\definecolor{coloract}{rgb}{0, 0.3, 0}
\definecolor{colorcact}{rgb}{.5, 0.3, 0.0}
\definecolor{colorclock}{rgb}{0.4, 0, 0}
\definecolor{blueColorBlind} {RGB}{68 ,119,170}
\definecolor{greenColorBlind}{RGB}{34 ,136,51 }
\definecolor{redColorBlind}  {RGB}{238,102,119}

\usepackage{graphicx}

\newcommand{\fakeParagraph}[1]{\smallskip

\noindent\textbf{\textsf{#1.}}}

\newcommand{\defProblem}[3]
{%
	\noindent\fcolorbox{black}{blueColorBlind!10}{
	\begin{minipage}{.95\columnwidth}
		\textbf{#1 problem:}\\
		\textsc{Input}: #2\\
		\textsc{Problem}: #3
	\end{minipage}
}

	\smallskip

}

\newif\ifinfigure
\let\infigure\iffalse %

\let\oldtikzpicture\tikzpicture
\let\endoldtikzpicture\endtikzpicture
\renewenvironment{tikzpicture}
  {\let\ifinfigure\iftrue\oldtikzpicture}
  {\endoldtikzpicture\let\ifinfigure\iffalse}

\newcommand{\styleact}[1]{\ensuremath{
	\ifinfigure%
		\textcolor{coloract}{{#1}}%
	\else%
		#1%
	\fi%
	}%
}

\newcommand{\stylecact}[1]{\ensuremath{
	\ifinfigure%
		\textcolor{colorcact}{{#1}}%
	\else%
		#1%
	\fi%
	}%
}

\newcommand{\styleclock}[1]{\ensuremath{
	\ifinfigure%
		\textcolor{colorclock}{{#1}}%
	\else%
		#1%
	\fi%
	}%
}

\newcommand{\styleloc}[1]{\ensuremath{
	\ifinfigure%
		\textcolor{black}{{\mathrm{#1}}}%
	\else%
		#1%
	\fi%
	}%
}

\newcommand{\assign}{\leftarrow}

\newcommand{\checkUseMacro}[1]{#1}

\newcommand{\set}[1]{\ensuremath{\left\{#1\right\}}}

\newcommand{\setN}{\ensuremath{\mathbb{N}}}
\newcommand{\setQ}{\ensuremath{\mathbb{Q}}}
\newcommand{\setQgeqzero}{\ensuremath{\setQ_{\geq 0}}}
\newcommand{\setR}{\ensuremath{\mathbb{R}}}
\newcommand{\setRgeqzero}{\ensuremath{\setR_{\geq 0}}}

\newcommand{\setZ}{\ensuremath{\mathbb{Z}}}
\newcommand{\setRpositif}{\ensuremath{\setR_{>0}}}

\newcommand{\BTrue}{\ensuremath{\mathit{true}}}

\newcommand{\compOp}{\bowtie}

\newcommand{\intpart}[1]{\ensuremath{\lfloor#1\rfloor}}
\newcommand{\fract}[1]{\ensuremath{\text{frac}(#1)}}

\newcommand{\init}{\ensuremath{0}}
\newcommand{\priv}{\ensuremath{{\mathit{priv}}}}
\newcommand{\pub}{\ensuremath{{\mathit{pub}}}}
\newcommand{\final}{\ensuremath{f}}

\newcommand{\styleAutomaton}[1]{\ensuremath{\mathcal{#1}}}

\newcommand{\styleSetRegion}[1]{\ensuremath{\mathsf{#1}}}
\newcommand{\styleBigSet}[1]{\ensuremath{\mathbb{#1}}}

\newcommand{\clock}{\ensuremath{\styleclock{x}}}
\newcommand{\constantmax}[1]{c_{#1}}

\newcommand{\clocky}{\ensuremath{\styleclock{y}}}
\newcommand{\clockz}{\ensuremath{\styleclock{z}}}
\newcommand{\clocki}[1]{\ensuremath{\styleclock{\clock_{#1}}}}
\newcommand{\ClockCard}{H} %
\newcommand{\clockval}{\ensuremath{\mu}}
\newcommand{\ClockSet}{\ensuremath{\mathbb{X}}} %
\newcommand{\ClocksZero}{\ensuremath{\vec{0}}}
\newcommand{\tickclock}{\ensuremath{\styleclock{\clockz}}}

\newcommand{\resets}{\ensuremath{R}}
\newcommand{\reset}[2]{\ensuremath{[#1]_{#2}}}

\newcommand{\clockeq}{\approx}
\newcommand{\TA}{\ensuremath{\checkUseMacro{\styleAutomaton{A}}}}
\newcommand{\TB}{\ensuremath{\checkUseMacro{\styleAutomaton{B}}}} %

\newcommand{\TAFullWeak}{\ensuremath{\TA_{\textit{full} \rightarrow \textit{weak}}}}
\newcommand{\TAWeakFull}{\ensuremath{\TA_{\textit{weak} \rightarrow \textit{full}}}}

\newcommand{\action}{\ensuremath{\styleact{a}}}
\newcommand{\actiona}{\ensuremath{\styleact{a}}}
\newcommand{\actionb}{\ensuremath{\styleact{b}}}
\newcommand{\actionc}{\ensuremath{\styleact{c}}}
\newcommand{\actionend}{\ensuremath{\styleact{e}}}

\newcommand{\caction}{\ensuremath{\stylecact{\kappa}}} %
\newcommand{\cactioni}[1]{\ensuremath{\stylecact{\caction_{#1}}}}
\newcommand{\ActionSet}{\ensuremath{\Sigma}}

\newcommand{\extraAction}{\ensuremath{\styleact{\sharp}}}

\newcommand{\constraint}{\ensuremath{C}}

\newcommand{\edge}{\ensuremath{\checkUseMacro{e}}}
\newcommand{\edgei}[1]{\ensuremath{\checkUseMacro{\edge_{#1}}}}
\newcommand{\EdgeSet}{\ensuremath{E}}

\newcommand{\guard}{\ensuremath{g}}

\newcommand{\invariant}{\ensuremath{I}} %

\newcommand{\loc}{\ensuremath{\styleloc{\ell}}}
\newcommand{\loci}[1]{\ensuremath{\styleloc{\loc_{#1}}}}
\newcommand{\locinit}{\ensuremath{\styleloc{\loc_\init}}}
\newcommand{\locfinal}{\ensuremath{\styleloc{\loc_\final}}}
\newcommand{\FinalSet}{\ensuremath{L_f}}
\newcommand{\PrivSet}{\ensuremath{L_{\priv}}}
\newcommand{\locpriv}{\ensuremath{\styleloc{\loc_\priv}}}

\newcommand{\loctest}{\ensuremath{\styleloc{\loc^\mathit{test}}}}
\newcommand{\LocSet}{\ensuremath{L}}
\newcommand{\LocSetPub}{\ensuremath{\LocSet_{\pub}}}
\newcommand{\LocSetPriv}{\ensuremath{\LocSet_{\priv}}}

\newcommand{\longuefleche}[1]{\stackrel{#1}{\longrightarrow}}

\newcommand{\run}{\checkUseMacro{\rho}} %

\newcommand{\duration}{\ensuremath{\mathit{dur}}} %

\newcommand{\TimedWords}[1]{\ensuremath{\textit{TW}(#1)}}
\newcommand{\Language}{\ensuremath{\mathcal{L}}}

\newcommand{\machine}{\ensuremath{\mathcal{M}}}
\newcommand{\zerocell}{\textbf{0}}
\newcommand{\onecell}{\textbf{1}}
\newcommand{\machinealph}{\ensuremath{\set{\zerocell, \onecell}}}
\newcommand{\StatesExists}{\ensuremath{Q_{\exists}}}
\newcommand{\StatesForall}{\ensuremath{Q_{\forall}}}
\newcommand{\cur}{\ensuremath{\mathit{cur}}}

\newcommand{\region}{\ensuremath{r}}
\newcommand{\regioni}[1]{\ensuremath{\region_{#1}}}

\newcommand{\laststate}{\ensuremath{\mathit{last}}}
\newcommand{\firststate}{\ensuremath{\mathit{first}}}

\newcommand{\regset}[1]{\styleSetRegion{R}_{#1}}
\newcommand{\regaut}[1]{\styleAutomaton{R}_{#1}}

\newcommand{\silentaction}{\ensuremath{\styleact{\varepsilon}}}

\newcommand{\RegAutTransitions}{\delta^{\styleSetRegion{R}}}

\newcommand{\belset}[1]{\ensuremath{Q_{#1}}}
\newcommand{\belaut}[2]{\styleAutomaton{B}_{#1}^{#2}} %
\newcommand{\belief}{\ensuremath{B}}
\newcommand{\belalpha}{\ensuremath{\Gamma}}
\newcommand{\belinit}{\ensuremath{\belief_0}}
\newcommand{\beltrans}{\ensuremath{\Delta}}
\newcommand{\localstrat}{\ensuremath{S}}
\newcommand{\controlflag}{\mathit{c}}
\newcommand{\nocontrolflag}{\mathit{\overline{c}}}

\newcommand{\strategy}{\ensuremath{\sigma}}
\newcommand{\strategyBlocksAll}{\ensuremath{\strategy_\emptyset}}
\newcommand{\buffered}[1]{\ensuremath{\mathit{buff}(#1)}}
\newcommand{\TAcontrolextend}{\ensuremath{\left(\ObsSet, \ControlSet, \LocSet, \locinit, \PrivSet, \FinalSet, \ClockSet, \invariant, \EdgeSet\right)}}
\newcommand{\clabel}[2]{\ensuremath{\frac{#1}{#2}}}
\newcommand{\ObsSet}{\ensuremath{\ActionSet_o}}
\newcommand{\ControlSet}{\ensuremath{\ActionSet_c}}
\newcommand{\ActionControlSet}{\ensuremath{\ActionSet^\times}}

\newcommand{\reachcondition}{non-blocking}

\newcommand{\Problem}{$\exists$SO}
\newcommand{\ProblemRC}{$\exists$NBSO}
\newcommand{\ProblemNSS}{$\exists N$-SSO}
\newcommand{\ProblemNSSRC}{$\exists$NB$N$-SSO}
\newcommand{\ProblemOSS}{$\exists$OSSO}
\newcommand{\ProblemOSSRC}{$\exists$NBOSSO}

\newcommand{\Succ}{\ensuremath{\mathit{Succ}}}
\newcommand{\Next}{\ensuremath{\mathit{Next}}}

\newcommand{\TTS}{\ensuremath{\mathfrak{T}}}

\newcommand{\TTSstate}{\ensuremath{s}}

\newcommand{\TTStransition}{\ensuremath{\transition}}

\newcommand{\semantics}[1]{\ensuremath{\mathfrak{T}_{#1}}}
\newcommand{\semanticscontrolextend}{\ensuremath{\left(\StateSet, \concstateinit, \ActionControlSet \cup \setRgeqzero, \transition\right)}}

\newcommand{\transition}{{\ensuremath{\rightarrow}}}

\newcommand{\StateSet}{\ensuremath{\mathfrak{S}}}
\newcommand{\concstate}{\ensuremath{\mathfrak{s}}}

\newcommand{\concstateinit}{\ensuremath{\concstate_\init}}

\newcommand{\AllRuns}{\ensuremath{\Omega}}

\newcommand{\CRuns}[2]{\ensuremath{\AllRuns^{#2}(#1)}} %

\newcommand{\PrivRuns}[2]{\ensuremath{\AllRuns^{#2}_{\mathit{priv}}(#1)}} %
\newcommand{\PubRuns}[2]{\ensuremath{\AllRuns^{#2}_{\mathit{pub}}(#1)}} %

\newcommand{\textstyleparam}[1]{\ensuremath{{#1}}}
\newcommand{\paramd}{\textstyleparam{\ensuremath{d}}}

\newcommand{\tw}{\ensuremath{\mathit{tw}}}
\newcommand{\clabeli}[1]{\ensuremath{\clabel{\action_{#1}}{\caction_{#1}}}}
\newcommand{\act}[1]{\ensuremath{\mathit{act}(#1)}}
\newcommand{\ctr}[1]{\ensuremath{\mathit{ctr}(#1)}}
\newcommand{\uncontrol}{\ensuremath{\stylecact{u}}}

\newcommand{\PrivateTr}[2]{\ensuremath{\mathit{PrivTr}^{#2}(#1)}} %
\newcommand{\PublicTr}[2]{\ensuremath{\mathit{PubTr}^{#2}(#1)}} %
\newcommand{\Trace}[1]{\ensuremath{\mathit{Tr}(#1)}}
\newcommand{\PartialTraces}[1]{\ensuremath{\mathit{PartialTr}(#1)}}
\newcommand{\PTrace}[2]{\ensuremath{\mathit{Tr}_{#2}(#1)}} %
\newcommand{\Parts}[1]{\ensuremath{\mathcal{P}(#1)}}

\newcommand{\APriv}{\ensuremath{\TA_{\mathit{priv}}}}
\newcommand{\APub}{\ensuremath{\TA_{\mathit{pub}}}}

\newcommand{\interval}[2]{\ensuremath{[\![#1 ; #2]\!]}}
\newcommand{\treg}[2]{\ensuremath{\left(#1 ; #2\right)}}

\newcommand{\TimeRegions}{\ensuremath{\mathit{TReg}}}
\newcommand{\TRegion}[1]{\ensuremath{\mathit{reg}(#1)}}

\newcommand{\restrict}[2]{\ensuremath{\langle #1 \rangle_{#2}}}

\newcommand{\class}[1]{[#1]}

\newcommand{\beliefcontrol}[2]{\ensuremath{\mathfrak{b}}_{#1}^{#2}}
\newcommand{\beliefcontrolset}[2]{\styleBigSet{B}_{#1}^{#2}}

\newcommand{\finalclass}[1]{\styleSetRegion{R}^F_{#1}}
\newcommand{\secret}[1]{\styleSetRegion{Private}_{#1}}
\newcommand{\notsecret}[1]{\styleSetRegion{Public}_{#1}}

\newcommand{\compatible}{\ensuremath{\strategy}\mbox{-compatible}}

\newcommand{\substrat}{\ensuremath{\nu}}

\newcommand{\transitionWith}[1]{\stackrel{#1}{\mapsto}}

\newcommand{\LocFinalSet}{F}

\newcommand{\untimed}[1]{\ensuremath{\mathit{unt}(#1)}}
\newcommand{\timed}[1]{\ensuremath{\mathit{time}(#1)}}
\newcommand{\tick}{\ensuremath{\blacktriangleright}}
\newcommand{\notick}{\ensuremath{\vartriangleright}}
\newcommand{\runend}{\ensuremath{\bigtriangledown}}

\newcommand{\ComplexityFont}[1]{{\sffamily\upshape #1}}
\newcommand{\PTIME}{\ComplexityFont{PTIME}\xspace}

\newcommand{\TwoEXPTIME}{\ComplexityFont{2EXPTIME}\xspace}
\newcommand{\ThreeEXPTIME}{\ComplexityFont{3EXPTIME}\xspace}

\newcommand{\PSPACE}{\ComplexityFont{PSPACE}\xspace}

\newcommand{\EXPSPACE}{\ComplexityFont{EXPSPACE}\xspace}

\newcommand{\eg}{e.g.,\xspace}

\newcommand{\ie}{i.e.,\xspace}

\newcommand{\wlogen}{\arXivVersion{without loss of generality}\CONCURVersion{w.l.o.g.\xspace}}

\title{Buffered control for opacity in timed automata} %

\author{\'Etienne André} {Nantes Université, CNRS, LS2N, Nantes, France \and Institut universitaire de France (IUF) \and \url{https://lipn.univ-paris13.fr/~andre/}}{}{https://orcid.org/0000-0001-8473-9555}{}

\author{Sarah Dépernet} {Université de Lorraine, CNRS, Inria, LORIA, F-54000 Nancy, France}{}{https://orcid.org/0009-0003-8710-7934}{}

\author{Engel Lefaucheux} {Université de Lorraine, CNRS, Inria, LORIA, F-54000 Nancy, France \and \url{https://elefauch.github.io/}}{}{https://orcid.org/0000-0003-0875-300X}{}

\authorrunning{\'E.\ André, S.\ Dépernet and E.\ Lefaucheux} %

\Copyright{\'Etienne André, Sarah Dépernet and Engel Lefaucheux} %



\category{} %

\CONCURVersion{
	\relatedversion{} %
}
\arXivVersion{
	\relatedversion{The current manuscript is the extended version of the manuscript of the same name published in the proceedings of the 37th International Conference on Concurrency Theory (CONCUR 2026); the current manuscript notably includes all proofs.} %
}

\funding{This work is partially supported by ANR BisoUS (ANR-22-CE48-0012) and ANR GUMMIS (ANR-25-CE48-5096).}%

\acknowledgements{We would like to thank the CONCUR reviewers for useful suggestions of improvement.}%

\nolinenumbers %

\EventEditors{John Q. Open and Joan R. Access}
\EventNoEds{2}
\EventLongTitle{42nd Conference on Very Important Topics (CVIT 2016)}
\EventShortTitle{CVIT 2026}
\EventAcronym{CVIT}
\EventYear{2026}
\EventDate{December 24--27, 2016}
\EventLocation{Little Whinging, United Kingdom}
\EventLogo{}
\SeriesVolume{42}
\ArticleNo{1}

\begin{document}
\sloppy

\maketitle              %

\begin{abstract}
Timed automata are an extension of finite automata that can measure and react to the passage of time, handling real-time constraints by using clocks.
The timed opacity problem, where an attacker attempts to infer from observed actions and timestamps whether a secret location was visited, was shown undecidable for timed automata.
Execution-time opacity is a decidable though limited setting in which the attacker attempts to detect whether the secret location was visited, by only relying on the run duration.
Here, we significantly extend this setting, by allowing the attacker to observe all observable actions, in the right order though with only the integral parts of their timestamps, which we call buffered observations.
We consider the controlled setting, in which we aim at dynamically defining a sequence of sets of enabled actions ensuring opacity with buffered observations.
We first prove the inter-reducibility of full opacity (observations must not leak the visit of the secret location) and weak opacity (the attacker might prove that the location was not visited, but not that it was visited) in this new controlled setting.
Then, we prove the undecidability of the problem of existence of a sequential control strategy ensuring opacity under buffered observations.
Finally and most importantly, we prove that decidability is retrieved in two independent cases, with their tight theoretical complexities, with and without control.
These two assumptions express realistic limitations of the controller.
The first case is when the strategy of the controller changes at most an \emph{a priori} fixed number of times per time unit, which is not a strong practical assumption.
The second case is when all controllable actions are observable and distinguishable by an attacker.
\end{abstract}

\newpage

\section{Introduction}

Information‑system security remains a cornerstone of modern computing, yet its challenges become markedly more intricate when timing constraints are introduced.
In many safety‑critical and real‑time applications—such as embedded controllers, network protocols, and cyber‑physical systems—the correctness of security mechanisms depends not only on what actions occur but also on when they occur.
Timed automata (TAs)~\cite{AD94} provide a natural and expressive formalism for modelling such time‑dependent behaviours, allowing to capture both discrete transitions and continuous clock evolutions within a unified framework.
Within this timed setting, opacity has emerged as a fundamental privacy property: a system is opaque with respect to a secret if an external observer, having full or partial visibility of the system's observable actions and timestamps, can never infer a given secret, typically the reachability of a given location. This information can be hidden by ensuring that each observation associated to a secret behaviour is also produced by a public behaviour.
Opacity thus formalizes the notion of ``information hiding'' against attackers who exploit timing information.

It is well-known since Cassez' seminal work on timed opacity~\cite{Cassez09} that deciding opacity for TAs is undecidable in dense time when the attacker observes the full timed word produced by a run, even for restricted subclasses of TAs such as event-recording automata~\cite{AFH99}.
\arXivVersion{%
	The interplay between dense‑time clocks and nondeterministic choices yields an infinite state space that cannot be effectively explored by any algorithmic procedure.
}%
This undecidability result underscores the inherent difficulty of guaranteeing privacy in real‑time systems and motivates the development of approximate or restricted analyses.
This problem becomes however decidable in a discrete-time setting~\cite{AGWZH24,ADL26,KKG24}, where actions only occur at integral times.
However, discrete-time analyses can be limited, and we address here the general dense-time case.

\begin{figure}[tb]
	\centering

	\begin{tikzpicture}[x=4cm, y=1cm]

	\draw[->] (-0.1,0) -- (2.3,0) node[right] {$t$};

	\foreach \x in {0,1,2}{
		\draw[very thick] (\x,0.2) -- (\x,-0.2);
		\node[below] at (\x, -0.2) {$\x$};
	}

	\draw (0,0.15) -- (0,-0.15);
	\node[above] at (0,0.15) {$\actiona$};
	\draw (0.2,0.15) -- (0.2,-0.15);
	\node[above] at (0.2,0.15) {$\actionb$};
	\node[below] at (0.2,-0.15) {$0.2$};

	\draw (0.8,0.15) -- (0.8,-0.15);
	\node[above] at (0.8,0.15) {$\actionc$};
	\node[below] at (0.8,-0.15) {$0.8$};

	\draw (1.2,0.15) -- (1.2,-0.15);
	\node[above] at (1.2,0.15) {$\actionb$};
	\node[below] at (1.2,-0.15) {$1.2$};

	\draw[dashed] (0.2, -.6) -- (0.2,-.75) -- (1.2,-.75) -- (1.2, -.6);

	\end{tikzpicture}
	\caption{Example of observation with infinite precision}
	\label{figure:buffer}
\end{figure}
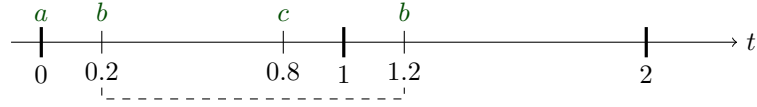

\fakeParagraph{Contributions}
In this paper (as in a number of former works), the secret consists in the visit (or not) of a given set of private locations.
The attacker aims at deciding whether at least one such private location was visited, by observing the system traces; the attacker additionally knows the system model (here, a~TA).
In practice, an attacker may not have infinite precision in its observations.
Therefore, it is interesting to study timed opacity in the setting of \emph{buffered observations}: instead of the exact timed word as in~\cite{Cassez09}, the attacker observes all actions within each one-unit time interval, in the right order but without their precise timestamp.
The analogy is therefore a \emph{buffer} that records all observable actions, and that is read every time unit.
This is particularly realistic when the attacker does not have the ability (\eg{} in terms of memory) to record all exact timestamps.
For example, consider the sequence of actions in \cref{figure:buffer}.
An attacker with infinite precision in its observations will see that the two occurrences of $\actionb$ are separated by exactly 1 time unit, which may give crucial information towards the knowledge of a secret.
However, under buffered observations, the attacker will only be able to see that $\actiona$ occurs at $t=0$, $\actionb$ and~$\actionc$ occur (in this order) between 0 and~1, while $\actionb$ occurs within $(1,2)$---in which case the aforementioned observation that both $\actionb$ are separated by 1 time unit is lost.
In fact, this behaviour becomes indistinguishable from another behaviour where the second $\actionb$ would occur, \eg{} at $t=1.1$ or at~$t=1.3$.

In the setting of buffered observations, opacity has been shown decidable~\cite{DQY25} (under the name of ``timed opacity against intruders with discrete-time precision'').
We extend to this setting of buffered observations the notion of \emph{control}: our aim is to decide the \emph{existence of a strategy}, \ie{} a dynamic selection of controllable actions such that the system is opaque.
This is of utmost importance when we aim at making a real-time system insensitive to timing attacks with a finite precision.

In this sense, our main contribution is twofold:
	we first prove the undecidability of the control problem with buffered observations.
	Second, we provide two assumptions expressing realistic limitations of the controller, each of them allowing to regain the decidability separately.
	The first assumption is that the strategy of the controller changes at most an \emph{a priori} fixed number of times per time unit.
	The second assumption is that all controllable actions are observable and distinguishable by an attacker.
Moreover, in each result of this paper, we address both opacity levels: \emph{weak} opacity (in which we only hide having visited a private location); and \emph{full} opacity (in which we additionally hide \emph{not} having visited it).
We show all of our proofs work for both flavours,
	by proving in the setting of control the inter-reducibility between weak and full opacity%
.

\fakeParagraph{Related work}
The aforementioned negative result from~\cite{Cassez09} leaves hope only if the definition or the setting is changed\arXivVersion{, which was done in four main lines of work}.
First, in~\cite{WZ18,WZA18}, the input model is simplified to \emph{real-time automata}\arXivVersion{~\cite{Dima01}, a restricted formalism compared to TAs: real-time automata can be seen as TAs with a single clock, reset at each transition}.
\cite{LLHL22,LLHL25}~exhibits decidability results for constant-time labelled automata\arXivVersion{, a subclass of real-time automata where events occur at constant values.
In this setting, initial-state opacity (``according to the observations, what was the initial state?'')\ and current-state opacity (``according to the observations, what is the current location?'')\ become decidable}.
In~\cite{Zhang24}, Zhang studies decidability for labelled real-time automata\arXivVersion{ (a subclass of labelled TAs); in this setting, state-based (at the initial time, the current time, etc.)\ opacity is proved to be decidable by extending the observer (that is, the classical powerset construction) from finite automata to labelled real-time automata}.

Second, in~\CONCURVersion{\cite{AGWZH24,ADL26,KKG24,DQYang25,AALL26}}\arXivVersion{\cite{AGWZH24,ADL26,DQYang25,AALL26}}, decidability results are exhibited in the setting of Cassez' definition, but with restrictions in the model: one-clock automata, one-action automata, event-recording automata, over discrete time or equivalently with integer resets.
\arXivVersion{%
	Similarly, in~\cite{KKG24}, discrete-time automata with several clocks are considered and transformed into tick automata in order to verify the current-state opacity.
	The discrete time setting yields decidability, while restricting the number of actions to~1 preserves undecidability; for a single clock, decidability can only be envisioned without silent actions~\cite{ADL26} (allowing silent actions or allowing two clocks immediately leads to undecidability).
}%

Third, in~\cite{AEYM21,AAL24}, the authors consider a \emph{time-bounded} notion of the opacity of~\cite{Cassez09}, where the attacker has to disclose the secret before a deadline, using a partial observability.
\arXivVersion{%
	This can be seen as a secrecy with an \emph{expiration date}.
	The rationale is that retrieving a secret ``too late'' is useless; this is understandable, \eg{} when the secret is the value in a cache; if the cache has been overwritten since, then knowing the secret is probably useless in most situations.
	In addition, the analysis is carried over a time-bounded horizon; this means there are two time bounds in~\cite{AEYM21}: one for the secret expiration date, and one for the bounded-time execution of the system.
Deciding opacity in this setting is shown to be decidable for~TAs.
}%
A somehow similar framework is considered in~\cite{SLR23}, in which the attacker has a bounded memory, and a finite duration between distinct observations is required, in which case the problem is decidable and \PSPACE{}-complete. 
\arXivVersion{The bounded memory side was also considered for related models such as time Petri net~\cite{DLL25}.}

Fourth, in~\cite{ALLMS23}, an alternative definition\arXivVersion{ to Cassez' opacity} is proposed, by studying execution-time opacity (ET-opacity): the attacker has only access to the \emph{execution time} of the system, as opposed to Cassez' partial observations where some events (with their timestamps) are observable.
The goal for the attacker is to deduce whether a special secret location was visited, by observing only the execution time.
In that case, most problems for TAs become decidable\arXivVersion{, including some problems when introducing an expiration date~\cite{ALM23} (see \cite{ALLMS23} for a survey)}.

Control of opacity was addressed in the simpler setting of ET-opacity\arXivVersion{, in which the attacker has only a single observation---that of the duration of an accepting run}:
in this setting, control was investigated with an untimed controller in~\cite{ABLM22} and a timed controller in~\cite{ADLL26}.
\CONCURVersion{%
	Control was also addressed in~\cite{BCLR15} in the setting of \emph{non-interference} in~TAs.
}\arXivVersion{%
	Regarding non-interference for TAs, some decidability results are proved in~\cite{BDST02,BT03}, while control was considered in~\cite{BCLR15} and a parametric timed extension in~\cite{AK20}.
}

General security problems for TAs are surveyed in~\cite{AA23survey}.

\fakeParagraph{Outline}
We recall necessary preliminaries, and introduce our setting of buffered observations and formal problem in \cref{sec:preliminaries}.
We then prove that full and weak opacity can be inter-reduced (\cref{sec:weak-full}).
We then prove our main contributions:
\arXivVersion{first, }controlling opacity under buffered observations is undecidable in general (\cref{sec:undecidable}),
but decidability can be retrieved \arXivVersion{using two assumptions, \ie{} }whenever the number of controller changes per time unit is statically bounded (\cref{sec:NSS}) or whenever all controllable actions are distinguishable by an attacker (\cref{sec:OSS}).

\section{Preliminaries}\label{sec:preliminaries}

We denote by $\setN, \setZ, \setQgeqzero, \setRgeqzero$ the sets of non-negative integers, integers, non-negative rationals and non-negative reals, respectively.
If $a$ and~$b$ are two integers with $a \leq b$, the set $\set{a, a+1, \dots, b-1, b}$ is denoted by $\interval{a}{b}$.

\emph{Clocks} are real-valued variables that all evolve over time at the same rate.
Throughout this paper, we assume a set~$\ClockSet = \{ \clocki{1}, \dots, \clocki{\ClockCard} \} $ of \emph{clocks}.
A \emph{clock valuation} is a function
$\clockval : \ClockSet \rightarrow \setRgeqzero$, assigning a non-negative value to each clock.
We write $\ClocksZero$ for the clock valuation assigning $0$ to all clocks.
Given a constant $d \in \setRgeqzero$, $\clockval + d$ denotes the valuation \st\ $(\clockval + d)(\clock) = \clockval(\clock) + d$, for all $\clock \in \ClockSet$. If $R$ is a subset of $\ClockSet$ and $\clockval$ a clock valuation, we call \emph{reset} of the clocks of $\resets$ and denote by $[\clockval]_R$ the valuation \st\ for all clock $x \in \ClockSet$, $\reset{\clockval}{\resets}(x) = 0$ if $x \in \resets$ and $\reset{\clockval}{\resets}(x) = \clockval(x)$ otherwise.

A \emph{constraint}~$\constraint$ is a conjunction of inequalities over~$\ClockSet $ of the form
$\clock \compOp d$, with
${\compOp} \in \{<, \leq, =, \geq, >\}$
and
$d \in \setZ$.
Given~$\constraint$, we write~$\clockval \models \constraint$ if the expression obtained by replacing each~$\clock$ with~$\clockval(\clock)$ in~$\constraint$ evaluates to true.

\subsection{Timed automata}\label{sec:TA}

Timed automata (TAs) extend finite automata with a finite set of clocks with values in~$\setRgeqzero$.
We extend standard TAs with
\begin{ienumerate}%
	\item a set of private locations (to model the secret), and
	\item \emph{two} action sets: one for observable events and one for controllable events.
\end{ienumerate}%
Each action is thus represented as a pair~$\clabel{\action}{\caction}$ specifying both its observable component~$\action$ and its controllable component~$\caction$.

\arXivVersion{\begin{definition}[Timed automaton (TA)]}
\CONCURVersion{\begin{definition}}
\label{def:TA}
	A TA~$\TA$ is a tuple \mbox{$\TA = \TAcontrolextend$}, where:
	\begin{ienumerate}%
		\item $\ObsSet$ is a finite set of observable actions;
		\item $\ControlSet$ is a finite set of controllable actions;
		\item $\LocSet$ is a finite set of locations,
		\item $\locinit \in \LocSet$ is the initial location,
		\item $\PrivSet \subseteq \LocSet$ is a set of private locations,
		\item $\FinalSet \subseteq \LocSet$ is a set of final locations,
		\item $\ClockSet$ is a finite set of clocks,
		\item $\invariant$ is the invariant, assigning to every $\loc\in \LocSet$ a constraint $\invariant(\loc)$, and
		\item $\EdgeSet$ is a finite set of edges  $\edge = (\loc,\guard,\clabel{\action}{\caction},\resets,\loc')$
		where~$\loc,\loc'\in \LocSet$ are the source and target locations, $\action \in \ObsSet \cup \{ \silentaction \}$ (where $\silentaction$ denotes an unobservable action), 
		$\caction \in \ControlSet \cup \{ \uncontrol \}$ (where $\uncontrol$ denotes an uncontrollable action),
		$\resets\subseteq \ClockSet$ is a set of clocks to be reset, and $\guard$ is a constraint %
			(called \emph{guard}).
	\end{ienumerate}%
\end{definition}
\begin{figure}[tb]

	\centering
	\large

	\begin{tikzpicture}[pta, scale=2, xscale=1, yscale=.5]

		\location[initial, name=s0, at={(0, -1)}]{$\styleclock{\clock} \leq 3$}{$\loci{0}$}

		\location[private, name=s2, at={(1.5, 0)}]{$\styleclock{\clock} \leq 2$}{$\locpriv$}

		\node[location, final] at (3, -1) (s1) {$\locfinal$};

		\path (s0) edge[bend left] node[above, align=center]{$\styleclock{\clock} \geq 1$\\$\clabel{\silentaction}{\cactioni{1}}$} (s2);
		\path (s0) edge[loop above] node[]{$\clabel{\action}{\cactioni{2}}$} (s0);
		\path (s0) edge[] node[below]{$\clabel{\actionb}{\uncontrol}$} (s1);
		\path (s2) edge[bend left] node[above]{$\clabel{\actionb}{\uncontrol}$} (s1);

	\end{tikzpicture}
	\caption{A TA example}
	\label{figure:example-TA}

\end{figure}
\begin{example}
	In \cref{figure:example-TA}, we give an example of a TA with three locations $\loci{0}$, $\locpriv$ and~$\locfinal$,
	four edges,
	two observable actions~$\{a, b\}$,
	two controllable actions~$\{ \cactioni{1}, \cactioni{2} \}$,
	and one clock~$\clock$.
	$\loci{0}$ is the initial location, $\locpriv$ is the (unique) private location, and~$\locfinal$ is the (unique) final location.
	$\loci{0}$ has an invariant ``$\styleclock{\clock} \leq 3$'' and the edge from $\loci{0}$ to $\locpriv$ is labelled by both the unobservable action~$\silentaction$ and the controllable action~$\cactioni{1}$, and has a guard ``$\styleclock{\clock} \geq 1$''.
\end{example}

\arXivVersion{
\paragraph{Semantics of timed automata}

We recall the semantics of a TA using a timed transition system~(TTS).
}

\begin{definition}[Semantics of a TA]\label{def:semantics}
	Given a TA $\TA = \TAcontrolextend$,
	the semantics of~$\TA$ is given by the timed transition system $\semantics{\TA} = \semanticscontrolextend$, with
	\begin{enumerate}
		\item $\StateSet = \big\{ (\loc, \clockval) \in \LocSet \times \setRgeqzero^\ClockSet \mid \clockval \models \invariant(\loc) \big\}$,
		\arXivVersion{\item }$\concstateinit = (\locinit, \ClocksZero) $,
		\item $\ActionControlSet = \set{\clabel{\action}{\caction} \mid (\action, \caction) \in (\ObsSet \cup \set{\silentaction})\times (\ControlSet \cup \set{\uncontrol})}$,
		\item  $\transition \subseteq (\StateSet \times \EdgeSet \times \StateSet) \cup (\StateSet \times \setRgeqzero \times \StateSet)$ consists of the discrete and (continuous) delay transition relations:
		\begin{enumerate}
			\item discrete transitions:
$\big((\loc,\clockval), \edge, (\loc',\clockval') \big) \in \transition$,
		if $(\loc, \clockval) , (\loc',\clockval') \in \StateSet$, $\edge = (\loc,\guard,\clabel{\action}{\caction},\resets,\loc') \in \EdgeSet$, $\clockval'= \reset{\clockval}{\resets}$, and $\clockval\models \guard$.
			\item delay transitions:
$\big((\loc,\clockval), d, (\loc,\clockval +d)\big) \in \transition$,
		if $d \in \setRgeqzero$ and
$\forall d' \in [0, d], (\loc, \clockval+d') \in \StateSet$.
		\end{enumerate}
	\end{enumerate}
\end{definition}

Moreover we write $(\loc, \clockval)\longuefleche{(d, \edge)} (\loc',\clockval')$ for a combination of a delay and discrete transition if
$\exists  \clockval'' :  \big( (\loc,\clockval), d, (\loc,\clockval'') \big) \in \transition$
and
$\big( (\loc,\clockval'') , \edge,  (\loc',\clockval') \big) \in \transition$. 

Given a TA~$\TA$ with semantics $\semanticscontrolextend$, we refer to the elements of~$\StateSet$ as the \emph{configurations} of~$\TA$.
A (finite) \emph{run} of~$\TA$ is an alternating sequence of configurations of~$\TA$ and pairs of delays and edges starting from the initial configuration $\concstateinit$ and ending in a final configuration (\ie{} whose location is final),
of the form
$(\loci{0}, \clockval_{0}), (d_0, \edge_0), (\loci{1}, \clockval_{1}), \ldots (\loci{n}, \clockval_{n})$ for some $n \in \setN$ (called the length of the run),
with
$\loci{n} \in \FinalSet$ and for $i = 0, 1, \dots n-1$, $\loci{i} \notin \FinalSet$, $\edge_i \in \EdgeSet$, $d_i \in \setRgeqzero, $ and
$(\loci{i}, \clockval_{i}) \longuefleche{(d_i, \edge_i)} (\loci{i+1}, \clockval_{i+1})$.
We denote by $\AllRuns(\TA)$ the set of runs of~$\TA$.
A \emph{path} of~$\TA$ is an infix of a run beginning and ending with a configuration. 

A common abstraction of the semantics of TAs is the \emph{region automaton}~\cite{AD94}, recalled in \cref{appendix:regions}.

\subsection{Buffered observations}\label{sec:timed-words}

A timed word is a sequence of pairs made up of a label (observation, control or both) and a timestamp in~$\setRgeqzero$, with the timestamps being non-decreasing over the sequence.
Timed words over an alphabet~$\ActionSet$ are therefore elements of $(\ActionSet \times \setRgeqzero)^*$.

We denote by $\TimedWords{\ActionSet}$ the set of all finite timed words over the alphabet~$\ActionSet$. A timed word over a set $\ActionControlSet$ of actions in $\ObsSet$ with control in $\ControlSet$ induces two timed words, one over the set of actions $\ObsSet$ and one over the set of controls $\ControlSet$, each in which the silent symbols $\silentaction$ and $\uncontrol$ have been removed. For such a timed word $w$ in $\TimedWords{\ActionControlSet}$, we denote by $\act{w}$ and $\ctr{w}$ the action and control parts of~$w$. For instance, if $w = (\clabel{\action}{\cactioni{2}}, 0.3) (\clabel{\actionb}{\cactioni{1}}, 0.8) (\clabel{\silentaction}{\cactioni{2}}, 2.6)(\clabel{\actionb}{\uncontrol}, 3.4)$, then $\act{w} = (a, 0.3) (b, 0.8) (b, 3.4)$ and $\ctr{w} = (\cactioni{2}, 0.3) (\cactioni{1}, 0.8) (\cactioni{2}, 2.6)$. We also denote the label and timestamp parts of the timed word (or timed sequence) $w$ by $\untimed{w}$ and $\timed{w}$ respectively. In the above example, $\untimed{w} = \clabel{\action}{\cactioni{2}} \clabel{\actionb}{\cactioni{1}} \clabel{\silentaction}{\cactioni{2}} \clabel{\actionb}{\uncontrol}$, $\timed{w} = 0.3 \cdot 0.8 \cdot 2.6 \cdot 3.4$; $\untimed{\act{w}} = abb$ and $\untimed{\ctr{w}} = \cactioni{2} \cactioni{1} \cactioni{2}$.

A run (or a path)~$\run$ of a TA~$\TA$ defines a timed word $\tw(\run)$ over $\ActionControlSet$: if $\run$ is of the form
\((\loci{0}, \clockval_0), (d_0, \edgei{0}), (\loci{1}, \clockval_1), \ldots, (\loci{n}, \clockval_n )\)
where for each  $i \in \interval{0}{n-1}$, $e_i = (\loc_i , g_i, \clabel{a_i}{\caction_i} , R_i, \loc_{i+1})$ and $\clabel{a_i}{\caction_i} \in \ActionControlSet$,
then it generates the timed word
$\tw(\run) = (\clabeli{j_0}, \sum\limits_{i=0}^{j_0}d_i)(\clabeli{j_1}, \sum\limits_{i=0}^{j_1}d_i)\cdots (\clabeli{j_m}, \sum\limits_{i=0}^{j_{m}}d_i)$, where $j_0 < j_1 < \dots < j_{m}$ and $\set{j_k \mid k \in \interval{0}{m}} = \set{i \in \interval{0}{n-1} \mid \clabeli{i} \neq \clabel{\silentaction}{\uncontrol}}$.

The set of timed words recognized by a TA~$\TA$ is the set of timed words generated by its runs,
\ie{} $\tw(\AllRuns(\TA))$, which we also denote $\Language(\TA)$ and call \emph{language} of~$\TA$.

A \emph{time region} is an interval either of the form $\set{n}$ or $(n; n+1)$ with $n \in \setN$.
The set of time regions is denoted by~$\TimeRegions$, and for $d \in \setRgeqzero$ then we denote by $\TRegion{d}$ the time region containing~$d$.
If $I$ is a convex union of time regions and $w$ a timed word, we define the \emph{restriction of $w$ to~$I$} as the maximal subword of~$w$ whose timestamps are all included in~$I$.
We denote it by~$\restrict{w}{I}$.
Formally, $\restrict{(\clabeli{0}, \tau_0)(\clabeli{1}, \tau_1) \dots (\clabeli{N}, \tau_N)}{I} = (\clabeli{k}, \tau_k) (\clabeli{k+1}, \tau_{k+1}) \dots (\clabeli{k+n}, \tau_{k+n})$ with for all $i$ in $\interval{0}{n}$, $\tau_{k+i} \in I$; if $k > 0$ then $\tau_{k-1} \notin I$ and if $k+n < N$ then $\tau_{k+n+1} \notin I$.
This restriction of a timed word to a convex union of time regions is also extended to runs: if $\run = (\loci{0}, \clockval_{0}), (d_0, \edge_0), (\loci{1}, \clockval_{1}), \ldots (\loci{N}, \clockval_{N})$ for some $N \in \setN$, the \emph{restriction of $\run$ to a convex union of regions $I$} is the path $\restrict{\run}{I} = (\loci{k}, \clockval_{k}), (d_k, \edge_k), (\loci{k+1}, \clockval_{k+1}), \ldots (\loci{k+n}, \clockval_{k+n})$ such that:
\begin{itemize}
\item if $k > 0$, $\sum\limits_{j=0}^{j=k-1} d_j \notin I$; and
\item for all $i \in \interval{0}{n}$, $\sum\limits_{j=0}^{j=k+i} d_j \in I$; and
\item if $k+n < N$, $\sum\limits_{j=0}^{k+n+1} d_j \notin I$.
\end{itemize}

We now define \emph{buffered observations}, which correspond to the exhaustive sequences (in the right order) of observable actions made during each time region and revealed, without their timestamp, when the next time region is reached.
New actions $\tick$ and~$\notick$ are introduced to identify the time region corresponding to each part of the buffered observation: a~$\notick$ is added in the buffered observation whenever the total execution time enters an open interval, and a~$\tick$ is added when the next integer is reached.
We also consider the end of an execution as an action and annotate it with the symbol $\runend$. 
If $w = (\clabeli{0}, \tau_0) (\clabeli{1}, \tau_1) \dots (\clabeli{n}, \tau_n)$ is a timed word over the set of controlled actions~$\ActionControlSet$, the buffered observation of~$w$, denoted by~$\buffered{w}$, is the untimed word
\[ \buffered{w} = \untimed{\restrict{\act{w}}{\set{0}}} \notick \untimed{\restrict{\act{w}}{(0;1)}} \tick \untimed{\restrict{\act{w}}{\set{1}}} \notick \dots \untimed{\restrict{\act{w}}{\TRegion{\tau_n}}}\runend z\]
for some $z \in \set{\tick, \notick}$.

\begin{example}
For a timed word $w = (\clabel{\action}{\cactioni{3}}, 0.3) (\clabel{\actionb}{\cactioni{1}}, 0.8) (\clabel{c}{\cactioni{2}}, 1)(\clabel{\silentaction}{\cactioni{3}}, 2.6)(\clabel{\actionb}{\uncontrol}, 3.5)$, the buffered observation is the untimed word $\notick a b \tick c \notick \tick \notick \tick \notick b \runend \tick$.
Each $\tick$ punctuates an integer global time and the symbol~$\notick$ indicates some time has elapsed since the last time unit, making a distinction between
\begin{ienumerate}%
	\item actions made during a time interval, and
	\item actions made at a precise integral time.
\end{ienumerate}%
This distinction gives exactly the time regions of all timestamps of the timed word. In particular, we also denote by $\TRegion{w}$ the time region of the end of the trace $w$.
In~$w$, $a$ and the first~$b$ occur during the time region $\treg{0}{1}$, $c$ appears at the precise time~$1$, while the last~$b$ occurs at a time in~$\treg{3}{4}$.
\end{example}

Throughout a path~$\run$, the controller or an attacker only has access to the buffered observation of this path, $\Trace{\run} = \buffered{\tw(\run)}$, called the \emph{trace} of~$\run$. We also use this notation for the traces of a set of runs.
A TA can be easily modified, so that all runs automatically include the symbols $\tick$, $\notick$ and~$\runend$, by adding a new clock reset every time unit, used to produce the $\tick$ and~$\notick$ (see \cref{sec:buffaut}).
That is, from now on, for every path $\run$,  $\Trace{\run} = \untimed{\tw(\run)}$.

\subsection{Control}\label{sec:def-strategies}

In~\cite{ADLL26}, a form of control was introduced in order to enforce execution-time opacity (``ET-opacity'').
There, a ``meta-strategy'' was restricting the set of transitions a run can take during each time region, and their order. In particular, this strategy was relying exclusively on the total amount of time elapsed since the beginning of the run, and fixing at each time unit its approximate behaviour until the next time unit.

In order to extend this control to the much broader setting of buffered observations, our controller may have access to some part of the trace during the run. We define sequential strategies, extending the notion of meta-strategy from~\cite{ADLL26} to our setting.
To do so, we first define the \emph{partial trace}~$\PTrace{\run}{I}$ of a run $\run$ in the time region~$I$ as the restriction of~$\Trace{\run}$ to the convex union of all time regions preceding~$I$, \ie{} to the interval~$[0; n)$ if~$I = \set{n}$ and to the interval~$[0; n]$ if~$I = (n; n+1)$.
This way, $\PTrace{\run}{I}$ gathers all observations taken into account by the controller at any time in $I$, when it announces its strategy for the time region $I$.
The set of all partial traces of the set $\AllRuns(\TA)$ of runs is denoted $\PartialTraces{\AllRuns(\TA)}$.

A \emph{sequential strategy}~$\strategy$ is a function mapping a partial trace to a finite sequence of subsets of the control alphabet~$\ControlSet$, such that if the time region is a singleton, then the sequence returned by the strategy includes only one element.
The intuition is as follows: at each time region, the strategy announces a sequence of control alphabets that will be used during that region.
If the current time region is a singleton~$\{n\}$ (i.e., an integer time), then only one set of enabled actions is needed, as the region has no duration.
If the current time region is an open interval~$(n; n+1)$, then the strategy may change the set of enabled actions multiple times during that interval, returning a sequence~$\Sigma_{n,1}\Sigma_{n,2}\ldots\Sigma_{n,K}$: the first set~$\Sigma_{n,1}$ is active from the beginning of the interval until the first change, $\Sigma_{n,2}$ from the first change until the second, and so on.
Formally:
\[
\begin{array}{rrcl}
	\strategy: & \PartialTraces{\AllRuns(\TA)} & \longrightarrow & \Parts{\ControlSet}^* \\
	& \PTrace{\run}{\set{n}} & \longmapsto & \ActionSet_{n,1} \textit{ if } \TRegion{\duration(\run)} = \set{n}\\
	& \PTrace{\run}{(n; n+1)} & \longmapsto & \ActionSet_{n,1} \ActionSet_{n,2} \dots \ActionSet_{n,K} \text{,if } \TRegion{\duration(\run)} = (n; n+1), \text{ for some } 0<K \in \setN
 
\end{array}
\]

For instance, in the TA of \cref{figure:example-TA}, a sequential strategy could allow only action~$\kappa_1$ during the time region~$(0;1)$ and then switch to allowing~$\kappa_2$ at the singleton region~$\{1\}$, giving the sequence~$\{\kappa_1\}$ for~$(0;1)$ and the single set~$\{\kappa_2\}$ for~$\{1\}$.

An \emph{implementation of a sequential strategy $\strategy$} extends $\strategy$ with exact timestamps.
That is, for each time region $I \in \set{\set{n}, (n;n+1)}$ and partial trace,
given the sequence $\ActionSet_{n,1} \ActionSet_{n,2} \dots \ActionSet_{n,K}$ defined by~$\strategy$,
an implementation $\substrat$ produces a timed sequence $(\ActionSet_{n,1}, t_{n,1}) (\ActionSet_{n,2}, t_{n,2}) \dots (\ActionSet_{n,K}, t_{n,K})$ with $(t_{n, k})_{k \in \interval{1}{K}}$ an increasing sequence with values in~$I$ and with $t_{n,K}=n+1$.
Formally, $\substrat$ \emph{implements} $\strategy$ if and only if for every buffered observation~$w$, the timed sequence $\substrat(w)$ is such that $\untimed{\substrat(w)} = \strategy(w)$ and $\restrict{\substrat(w)}{\TRegion{w}} = \substrat(w)$.

Let~$\substrat$ be some implementation of a strategy~$\strategy$. Let~$t \in \setRgeqzero$. The set of control labels \emph{enabled by~$\substrat$ at time~$t$ knowing the partial trace~$w$} is
\[ E_{\substrat}(t,w) = \strategy(w)_i \cup \set{\uncontrol} \; \text{ with } i = \min \set{k \mid t \leq \timed{\substrat(w)}_k }.\]

The runs of $\TA$ which transitions are all enabled by an implementation~$\substrat$ of the strategy~$\strategy$ are called \emph{$\strategy$-compatible} runs.
Their set is denoted $\CRuns{\TA}{\strategy}$.
Formally:%

\begin{align*}
\CRuns{\TA}{\strategy} &= \Big\{ \run \in \AllRuns(\TA) \mid \run = (\locinit, \clockval_{0})(d_0,\clabeli{0})\dots (d_{n-1},\clabeli{n-1})(\loci{n}, \clockval_{n}):
\exists \substrat \text{ implementing }\\
& \strategy \text{ such that } 
 \forall 0\leq i< n, \caction_i \text{ is enabled }\text{by } \substrat \text{ at time } \sum_{k=0}^i d_k \text{ knowing }\\& \PTrace{(\locinit, \clockval_{0})(d_0,\clabeli{0})\dots (d_{i-1},\clabeli{i-1})(\loci{i}, \clockval_{i})}{\TRegion{\sum_{k=0}^i d_k}}
\Big\}
\end{align*}

The trace of a $\strategy$-compatible run is called \emph{controlled trace}. The set of controlled traces $\Trace{\CRuns{\strategy}{\TA}}$ of a TA~$\TA$ with a strategy $\strategy$ is called the \emph{controlled language} of~$\TA$ under~$\strategy$.

We will consider two orthogonal restrictions of control with sequential strategies:
\begin{itemize}
\item A strategy for which every sequence of alphabets has length at most~$N$, for a given $N \in \setN$, is called an \emph{$N$-sequential strategy} ($N$-SS for short).
This restriction is natural, representing the inability of a control to make an arbitrarily high number of actions within a single time unit.
\item All controllable actions are supposed observable and distinguishable from each other 
and from uncontrollable actions. 
This assumption, called \emph{Observable Sequential 
Strategies} (OSS for short), directly concerns the \TA{} and the control alphabet: 
each control label~$\caction$ is associated to an observable action~$\action_\caction$ 
such that %
a transition has observation~$\action_\caction$ iff it has the control label~$\caction$.
\end{itemize}

\subsection{Opacity of a controlled TA}\label{sec:winning-condition}

Given a TA~$\TA$ and a run~$\run=(\loci{0}, \clockval_0), (d_0, \edgei{0}), (\loci{1}, \clockval_1), \ldots, (\loci{n}, \clockval_n)$ of~$\TA$, we say that $\run$ \emph{visits} the location $\loc \in \LocSet$ if there exists~$m \in \interval{0}{n}$ such that $\loci{m} = \loc$.
If $\run$ visits some $\locpriv \in \PrivSet$, we say it is a \emph{private run}; otherwise, it is a \emph{public run}.
We denote by~$\PrivRuns{\TA}{\strategy}$ (resp.\ $\PubRuns{\TA}{\strategy}$) the set of private runs (resp.\ public runs) under a strategy~$\strategy$.
We use $\PrivateTr{\TA}{\strategy} = \Trace{\PrivRuns{\TA}{\strategy}}$ to denote the set of traces of private runs of~$\TA$ under~$\strategy$, and $\PublicTr{\TA}{\strategy} = \Trace{\PubRuns{\TA}{\strategy}}$ for the set of traces of public runs of~$\TA$ under~$\strategy$.

A TA can be modified to store within locations whether the current run visited a private location or not (one only needs to include a Boolean in the locations that is turned to~$\BTrue$ if a private location was visited on the way to this location, see~\cite{ADL26}, illustrated in \cref{section:tools:Apriv-Apub}).
As such, we can assume \wlogen{} that the set of private locations $\PrivSet$ is absorbing (\emph{i.e.} no run can go from a location of $\PrivSet$ to a location of $\LocSet\setminus \PrivSet$). We denote $\LocSetPub=\LocSet\setminus \PrivSet$ the set of
public states (that by construction can only be reached by public runs).

Opacity is a notion to describe that, despite all information gathered by an attacker observing a run, they cannot deduce a secret (here, whether a private location was visited).
Thus we say that a trace is \emph{opaque under $\strategy$} if it can be obtained from both a public and a private run, \ie{} whenever it is in $\PrivateTr{\TA}{\strategy} \cap \PublicTr{\TA}{\strategy}$.
A TA~$\TA$ is \emph{weakly opaque under strategy $\strategy$} if $\PrivateTr{\TA}{\strategy} \subseteq \PublicTr{\TA}{\strategy}$ (all private traces are opaque under~$\strategy$) and \emph{fully opaque under~$\strategy$} if $\PrivateTr{\TA}{\strategy} =  \PublicTr{\TA}{\strategy}$ (all traces are opaque under~$\strategy$).

When considering a controlled TA,
a first option consists in exhibiting a strategy simply ensuring opacity.
A second option consists in requiring additionally that a final location remains reachable even with this restriction of the system's behaviour.
We say that the strategy~$\strategy$ is \emph{\reachcondition} on~$\TA$ if $\CRuns{\TA}{\strategy}$ is not empty.
This reachability condition is natural since in several cases, the best strategy achieving opacity is the one enabling as few actions as possible, or even entirely blocking the system.
Indeed, a TA with no possible run is opaque, and such a control strategy which is not \reachcondition{} is of little interest in practice.

We study the existence of a strategy ensuring opacity, both with and without requiring it to be \reachcondition{}.

\defProblem{Weak / Full \ \Problem{} / \ProblemRC{}: Existence of a (\reachcondition{}) sequential strategy making the TA fully opaque (resp.\ weakly opaque)}{A TA~$\TA$}{Does there exist a (\reachcondition{}) sequential strategy making~$\TA$ weakly (resp.\ fully) opaque?}

When limiting the search to $N$-sequential strategies or when in the observable sequential strategies settings, these problems are denoted:
\ProblemNSS{}, \ProblemNSSRC{}, \ProblemOSS{} and \ProblemOSSRC{}.

\begin{remark}
Opacity is worth investigating following two different approaches:
\begin{description}
	\item[offline] when the only traces we compare in order to decide opacity are those of runs reaching a final location;
	\item[online] when \emph{prefixes} are also considered, \ie{} traces of all paths eventually leading (or not) to a final location.
\end{description}

Here, we focus on \emph{offline opacity} and simply refer to it as \emph{opacity}\arXivVersion{ in the remainder of the paper}.
However, note that online opacity can be seen as a special case of offline opacity, since adding a silent transition (taken without any delay) from all locations to a final location amounts to add all paths to the set of runs reaching a final location.
As a consequence, solving online opacity problems can be done with the same algorithms as for offline opacity.
\end{remark}

\section{Inter-reducibility of weak and full opacity, with and without control}\label{sec:weak-full}

We first prove that full and weak opacity are inter-reducible in our setting, and 
thus that they have the same theoretical complexity.
In the subsequent sections, we will therefore only establish the results for one of the
two notions.

\begin{restatable}{theorem}{interreduc}
The weak and full variants of the 
\Problem{}, \ProblemRC{}, \ProblemNSS{}, \ProblemNSSRC{}, \ProblemOSS{} and \ProblemOSSRC{} problems are inter-reducible.
\end{restatable}

\begin{proof}[Proof sketch, see \cref{sec:interreduc} for the full proof]
	In~\cite{ADL26}, weak and full ET-opacity were shown to be inter-reducible in the absence of control.
	The proof mainly relied on the ability to construct from a TA~$\TA$ another TA~$\TA'$ where private and public runs were swapped.
	This proof cannot carry over to the controlled setting since there is no guarantee that the same strategy can be applied in both TAs to make them \emph{simultaneously} weakly opaque.

	Reducing weak opacity problems to full opacity problems is relatively straightforward, by reusing TA fragments from the proofs in~\cite{ADL26}.
	The opposite reduction is more technical: we use a single TA, and add a final letter (`$a$' or `$b$') pointing
	which inclusion is being tested: if the last letter of the trace is an `$a$', then the run participates in the test of
	the inclusion $\PrivateTr{\TA}{\strategy} \subseteq \PublicTr{\TA}{\strategy}$; and if the last letter is a `$b$', then the run participates in the converse inclusion.
\end{proof}

\section{Undecidability in the general case of buffered observations}\label{sec:undecidable}

While the buffering of the observations and the sequencing of the control weaken the required accuracy of the analysis, we show that the full \Problem{} and \ProblemRC{} problems are undecidable.
To do so, we proceed by reduction from the Post correspondence problem (PCP) to the problem of the existence of a sequential strategy ensuring full opacity, under buffered observations. 
Intuitively, we rely on the strategy having to plan an unbounded number of alphabets
over the next time interval to have it choose the full word of the PCP, and then the
following buffered observation reveals whether the choice was good (\ie{} ensured opacity).
Our reduction strongly relies on the absence of bounds on the number of control alphabets planned by the strategy in a single time interval (contrarily to the $N$-SS setting), and requires that the controllable actions do not appear in the buffered observation (contrarily to the OSS setting).

\begin{theorem}\label{theorem:general-undecidability}
The weak and full \Problem{} and \ProblemRC{} problems are undecidable.
\end{theorem}

\begin{proof}
Let $N \geq 2$ be an integer, and $V=(v^j)_{1 \leq j \leq N}$ and $W=(w^j)_{1 \leq j \leq N}$ be sequences of words over an alphabet $\ActionSet$.
The PCP asks whether there exists a sequence of indices $(i_k)_{1 \leq k \leq K}$, with $K \geq 1$ and for all $k$, $1 \leq i_k \leq N$, such that $v^{i_1} \cdot v^{i_2} \cdot \dots \cdot v^{i_K} = w^{i_1} \cdot w^{i_2} \cdot \dots \cdot w^{i_K}$. 
For $1 \leq i \leq N$, we denote by $n_i$ and $m_i$ the respective length of $v^i$ and~$w^i$.

We will build a TA~$\TB$ such that there exists a sequential strategy ensuring full opacity of $\TB$ iff the instance $(V,W)$ of the PCP has a solution starting with the pair $(v_1,w_1)$ (this starting assumption is without loss of generality as there are finitely many starting options).

The TA $\TB$ uses the control set $\ControlSet=\{1,\dots, N, \actionend\}$,
cannot take a transition labelled by a controllable action at time $0$ and thanks to an invariant on the final locations
every run reaches a final location before time~1.
Hence, every sequential strategy on~$\TB$ is defined entirely by the sequence of control alphabets it selects for the interval~$(0;1)$.
Moreover, $\TB$ is made of three components reached at a time greater than $0$ from the initial location.
The first two limit the shape of the strategies that  could ensure opacity, so that it can be associated to a word $i_1 \actionend i_2 \actionend\dots i_k \actionend$ where for each $j$, $1\leq i_j\leq N$, and a third which tests whether the sequence $i_1 i_2 \dots i_k $ is solution to the PCP instance.
More precisely, $\TB$ consists of:
\begin{enumerate}
\item a gadget ensuring that two actions cannot be selected at the same time (see \cref{fig:gadget-limit}), hence a strategy achieving opacity is associated to a single sequence of alphabets, each of which contains a single action (the strategy can thus be represented by a word $a_1\dots a_k\in \ControlSet^*$);%
\item a gadget which forces that the strategy alternates between actions $i\in \{1,\dots,N\}$ and the action $\actionend$, starts with an action $i\in \{1,\dots,N\}$ and ends with $\actionend$, giving the form claimed (see \cref{fig:gadget-order});
\item a part of the TA imitating the words of the PCP instance (see \cref{fig:Post}): mainly, if the strategy is associated to a word $i_1 \actionend i_2 \actionend\dots i_k \actionend$, a private run 
will read $v^{i_1} \cdot v^{i_2} \cdot \dots \cdot v^{i_K}$ and a public run will 
read $w^{i_1} \cdot w^{i_2} \cdot \dots \cdot w^{i_K}$.
\end{enumerate}

\begin{figure}[tb]
\begin{center}
\begin{tikzpicture}[pta, node distance=2cm, thin]

	\node[location, private, initial] (lt1) at (-1, 4.5) {$\loctest$};
	\node[location] (lt2) at (+1, 4.5) {$\loc^{a}$};
	\node[location,final] (lt4) at (+5, 4.5) {$\loc^{end}$};

	\path
    (lt1) edge node[align=center, above]{$\clabel{\silentaction}{\styleact{a}}$} node[below]{$\styleclock{\clock} \assign 0$} (lt2)
    (lt2) edge node[align=center]{$\clabel{\styleact{\natural}}{\styleact{b}}$} node[below]{$\styleclock{\clock} = 0$} (lt4);
\end{tikzpicture}
\end{center}
  \caption{Gadget forbidding selecting two actions simultaneously, with $a$ and~$b$ two actions in $\ControlSet$ with $b\neq a$ and $\natural$ an observation exclusive to this gadget.}
  \label{fig:gadget-limit}
\end{figure}
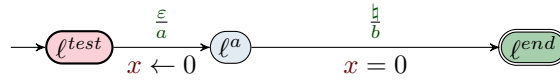

\fakeParagraph{First component}
Let us start by explaining the gadget of \cref{fig:gadget-limit}.
From an initial location $\loctest$, for every action $a$, if $a$ is allowed, then a
run can reach a location $\loc^a$ (there is thus one such location per controllable actions), while resetting a clock~$x$.
From $\loc^a$, if another action~$b$ is allowed
and if $x=0$ (in other words, if $b$ and~$a$ are allowed at the same time), the run can reach a final location while reading~$\natural$.
This thus produces a private run which trace is~$\natural$, and as $\natural$~is exclusive to this gadget, this run violates opacity.
Hence, any strategy that ensures opacity cannot allow two actions at the same time.
As by the invariant mentioned earlier a run can only reach a final location before time 1, and as no action can be done in the second part of the construction (more details on this below) at time~$0$, a sequential strategy can be defined exclusively by the sequence of control alphabet it selects for the interval~$(0;1)$. Moreover, due to the above,
if the strategy achieves opacity, only one action can be selected at a time. Thus
there is a correspondence between sequential strategies that could achieve opacity, and sequences of actions.

\begin{figure}[tb]
\begin{center}
\begin{tikzpicture}[pta, node distance=2cm, thin]

	\node[location,  initial] (lt1) at (-1, 4.5) {$\loc_0$};
	\node[location,private] (lt2) at (+1, 6) {$\locpriv$};
	\node[location,final] (lt4) at (+5, 4.5) {$\loc^{end}$};
	\node[location] (lp1) at (+1, 3.5) {$\loc_s$};
	\node[location] (lp2) at (+4, 3.5) {$\loc_i$};

	\path
    (lt1) edge node[align=center, below]{$\clabel{\silentaction}{\uncontrol}$} (lt2)
    (lt1) edge node[align=center, above]{$\clabel{\silentaction}{\uncontrol}$} (lp1)
    (lt2) edge[loop left] node[align=center, left]{$\clabel{\styleact{i}}
    {\styleact{i}}$} (lt2)
    (lp1) edge[bend left] node[align=center, above right]{$\clabel{\styleact{i}}
    {\styleact{i}}$} (lp2)
    (lp2) edge[bend left] node[align=center, above]{$\clabel{\silentaction}
    {\styleact{\actionend}}$} (lp1)
    (lt2) edge node[align=center]{$\clabel{\styleact{\flat}}{\uncontrol}$} (lt4)
    (lp1) edge[bend left] node[align=center]{$\clabel{\styleact{\flat}}{\uncontrol}$} (lt4)
    ;
\end{tikzpicture}
\end{center}
  \caption{Gadget limiting the order of the actions selected by the strategy, with $1\leq i\leq N$, and $\flat$ an observation exclusive to this gadget.}
  \label{fig:gadget-order}
\end{figure}
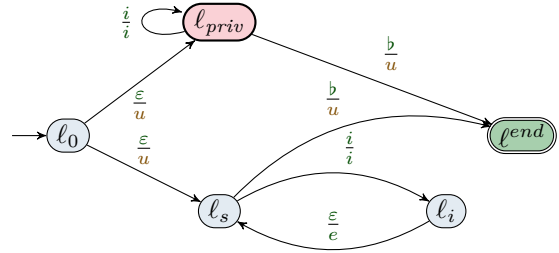

\fakeParagraph{Second component}
Let us explain the second gadget (\cref{fig:gadget-order}).
Thanks to the first gadget, a sequential strategy can be represented by a single word
$a_1\dots a_k\in \ControlSet^*$. 
In this gadget under this strategy, private runs will produce every word $w\sharp$
where $w$ is a subword of $a_1\dots a_k$ that does not contain~$\actionend$.
On the public side however, if a run reads a letter associated to an integer, it goes to 
the location $\loc_i$ from which it can exit only if an $\actionend$ is allowed. Hence, to copy
the private runs, public runs need every integer action to be followed by an~$\actionend$.
As $\flat$ cannot be produced by any other component of~$\TB$, this ensures the alternation claimed in the strategy.
While technically the strategy could start by allowing~$\actionend$, this will have no effect in
any of the components of $\TB$ and thus ignored without loss of generality.
In the following, we thus assume the sequential strategy is associated to a word $i_1 \actionend i_2 \actionend\dots i_k \actionend$.

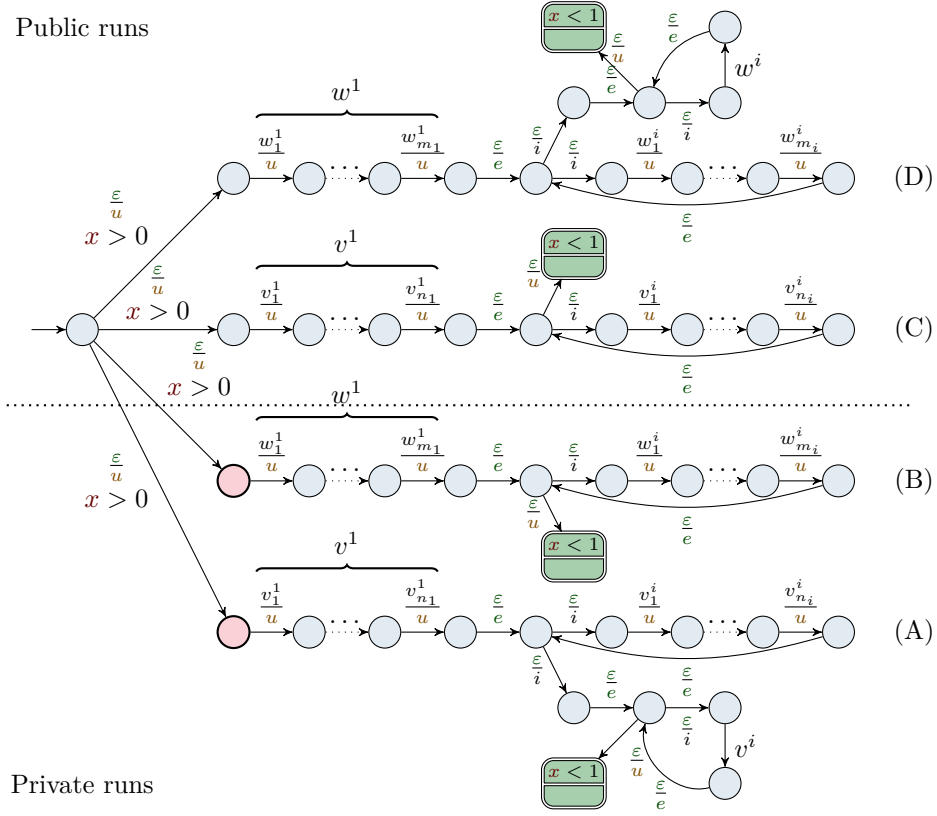
\begin{figure}[bt]
   \centering
	\begin{tikzpicture}[pta, node distance=2cm, thin]
	
	\node[location, initial] (l0) at (-1,1) {};
	
	\node[location, private] (A1) at (1, -3) {};
	\node[location, private] (B1) at (1, -1) {};
	\node[location] (C1) at (1, 1) {};
	\node[location] (D1) at (1, 3) {};
		
	\foreach \k in {2,3,...,9}
	{\node[location] (A\k) at (\k, -3) {};
	\node[location] (B\k) at (\k, -1) {};
	\node[location] (C\k) at (\k, 1) {};
	\node[location] (D\k) at (\k, 3) {};}

	\location[final, name=F1, at={(5.5,-2)}]{$\styleclock{\clock} < 1$}{}
	\location[final, name=F2, at={(5.5, 2)}]{$\styleclock{\clock} < 1$}{}
	\location[final, name=F3, at={(5.5,-5)}]{$\styleclock{\clock} < 1$}{}
	\location[final, name=F4, at={(5.5, 5)}]{$\styleclock{\clock} < 1$}{}

	\foreach \a/\p in {A/left, B/above right, C/above, D/above left}
	{\path (l0) edge[] node[\p, align= center] {$\clabel{\silentaction}{\uncontrol}$\\$\styleclock{x}>0$} (\a1);}
	
	\foreach \a/\u/\m in {A/v/n,B/w/m,C/v/n,D/w/m}
	{
	\path (\a1) edge[] node {$\clabel{\u^{1}_{1}}{\uncontrol}$} (\a2);
	\path (\a2) edge[dotted] node {$\dots$} (\a3);
	\path (\a3) edge[] node {$\clabel{\u^{1}_{\m_1}}{\uncontrol}$} (\a4);
	\path (\a4) edge[] node {$\clabel{\silentaction}{\actionend}$} (\a5);
	\path (\a5) edge[] node {$\clabel{\silentaction}{i}$} (\a6);
	\path (\a6) edge[] node {$\clabel{\u^{i}_{1}}{\uncontrol}$} (\a7);
	\path (\a7) edge[dotted] node {$\dots$} (\a8);
	\path (\a8) edge[] node {$\clabel{\u^{i}_{\m_{i}}}{\uncontrol}$} (\a9);
	\path (\a9) edge[bend left=16] node {$\clabel{\silentaction}{\actionend}$} (\a5);
	}

	\path (B5) edge[] node[left] {$\clabel{\silentaction}{\uncontrol}$} (F1);
	\path (C5) edge[] node {$\clabel{\silentaction}{\uncontrol}$} (F2);

	\node[location] (A10) at (5.5, -4) {};
	\node[location] (D10) at (5.5, 4) {};
	\node[location] (A11) at (6.5, -4) {};
	\node[location] (D11) at (6.5, 4) {};
	\node[location] (A12) at (7.5, -4) {};
	\node[location] (D12) at (7.5, 4) {};
	\node[location] (A13) at (7.5, -5) {};
	\node[location] (D13) at (7.5, 5) {};
	
	\foreach \a/\u in {A/v,D/w}
	{
	\path (\a5) edge[] node[left] {$\clabel{\silentaction}{i}$} (\a10);
	\path (\a10) edge[] node {$\clabel{\silentaction}{\actionend}$} (\a11);
	\path (\a11) edge[] node[below] {$\clabel{\silentaction}{i}$} (\a12);
	\path (\a12) edge[] node[right] {$\u^{i}$} (\a13);

	}

	\path (A13) edge[bend left=60] node[below] {$\clabel{\silentaction}{\actionend}$} (A11);
	\path (D13) edge[bend right=30] node[above] {$\clabel{\silentaction}{\actionend}$} (D11);

	\path (A11) edge[] node {$\clabel{\silentaction}{\uncontrol}$} (F3);
	\path (D11) edge[] node[above] {$\clabel{\silentaction}{\uncontrol}$} (F4);

	\draw[color=black, -, thick, decorate,decoration={brace,raise=0.1cm}]
(1.3,3.7) -- (3.7,3.7) node[above=0.2cm,pos=0.5, align=center] {$w^1$};
	\draw[color=black, -, thick, decorate,decoration={brace,raise=0.1cm}]
(1.3,1.7) -- (3.7,1.7) node[above=0.2cm,pos=0.5, align=center] {$v^1$};
\draw[color=black, -, thick, decorate,decoration={brace,raise=0.1cm}]
(1.3,-2.3) -- (3.7,-2.3) node[above=0.2cm,pos=0.5, align=center] {$v^1$};
	\draw[color=black, -, thick, decorate,decoration={brace,raise=0.1cm}]
(1.3,-0.3) -- (3.7,-0.3) node[above=0.2cm,pos=0.5, align=center] {$w^1$};

	\draw[dotted, -, thick] (-2,0) -- (10, 0);
	\node  at (-1, 5) {Public runs};
	\node  at (-1, -5) {Private runs};
	
	\node  at (10, 3) {(D)};
	\node  at (10, 1) {(C)};
	\node  at (10, -1) {(B)};
	\node  at (10, -3) {(A)};
	\end{tikzpicture}
	\caption{Reducing from PCP to \Problem{}}
  \label{fig:Post}
\end{figure}

\fakeParagraph{Third component}
We now explain the third component of the construction (\cref{fig:Post}).
This part of the TA consists in four branches reached at a time greater than~$0$.

First consider branch~B. A run going through this branch will first read the word $w^1$
and reach a location with many exiting transitions, allowing to select which word to read.
From there, either it takes the transition associated to~$i_1$ and reads word~$w^{i_1}$ before coming back to the starting location by a transition labelled by an~$\actionend$, or it will stay in that location.
This repeats for every action~$i_j$, before the run decides to
reach the final location (with an unobservable and uncontrollable action). As a consequence, this branch produces private runs with the buffered observation
$w_1 w^{i_{j_1}}w^{i_{j_r}}$ where  $j_1,\dots, j_r$ is an increasing subsequence of $1,\dots, k$. In other words, it starts by~$w_1$ and then selects which of the words allowed
by the strategy it reads.
Action~$\actionend$ ensures that, if an action~$i$ is allowed,
no run can take twice the transition associated to it, as an~$\actionend$ is necessary to
come back to the location where the selection occurs.

Branch~C is similar, producing public runs with the buffered observation
$v_1 v^{i_{j_1}}v^{i_{j_r}}$ where  $j_1,\dots, j_r$ is an increasing subsequence of $1,\dots, k$.

In particular, in order for full opacity to be ensured by the strategy, we see that, ignoring the other branches, the two longest words produced by runs in those branches
($w_1 w^{i_1}w^{i_k}$ and $v_1 v^{i_1}v^{i_k}$) must be equal as we want for the reduction.
However, the other subwords need to be equal as well, which is not required in the PCP.
This is the reason why the other two branches A and D exist.
Branch A acts like C, except that it produces private runs, and the buffered observations
it produces cannot be $v_1 v^{i_1}v^{i_k}$: in order to reach the final location, a run needs to take a transition with control $i$ and the following transition of control $\actionend$ to go to the second line of the branch. In other words, it needs to ``waste'' 
one opportunity to read a word. It thus can produce every buffered observation except
the full one.
Branch~D does the same for $w$ / as branch~B.

As a consequence, this component is fully opaque iff 
$w_1 w^{i_1}w^{i_k}=v_1 v^{i_1}v^{i_k}$, and thus iff the given PCP instance has a solution.

As a final location can be reached in the third component whatever the strategy does, 
this reduction applies to both the
full \Problem{} and full \ProblemRC{} problems.
\end{proof}

\section{Ensuring opacity with $N$-Sequential Strategies ($N$-SS)}\label{sec:NSS}

While the existence of a sequential strategy making the TA opaque under buffered observations is undecidable in general, we now prove that it becomes decidable when limiting the sequential strategies to \emph{sequences of a predefined given length}.

\begin{restatable}{theorem}{buffopaque}\label{theorem:NSS}
The decision problems \ProblemNSS{} and \ProblemNSSRC{} are \TwoEXPTIME-complete.
\end{restatable}

In the remaining of this section,
we construct our \emph{belief automaton}, an automaton that describes the set of possible regions the system can be in at each time under a given strategy.
We establish formally that this belief automaton correctly represents the behaviour of the controlled~TA.
The \ProblemNSS{} and \ProblemNSSRC{} problems can then be solved  through a game on the belief automaton.
We establish the \TwoEXPTIME-hardness of \ProblemNSS{} and \ProblemNSSRC{} by reduction from the halting problem for alternating Turing machines with exponential memory.
Detailed proofs of \cref{theorem:NSS} are in \cref{appendix:proof:NSS}.

We call \emph{belief} a set of regions\footnote{%
	We follow the vocabulary from, \eg{} \cite{BFHHH14}.
	This is also close to the concept of \emph{estimator} (\eg{} \cite{KKG24}).}.
In particular, we call \emph{natural belief} the set of regions in which the attacker \emph{believes} to be according to their knowledge, \ie{} the current buffered observation and the sequential strategy.
For a TA $\TA$ and a sequential strategy~$\strategy$, we denote by $\beliefcontrol{w}{\strategy}$ the set of regions in which the system can be after a buffered observation $w$ while following a strategy~$\substrat$ implementing $\strategy$ in~$\TA$, \ie{}
$r \in \beliefcontrol{w}{\strategy}$ iff there exists a run $\run$ in~$\TA$ such that $\run$ is \compatible,
 $\laststate(\run) \in r, r \in \regset{\TA}$ and $\tw(\run)=w$.
The set of all natural beliefs is denoted $\beliefcontrolset{\TA}{\strategy} = \{\beliefcontrol{w}{\strategy} \mid w\in B\Language(\TA,\strategy)\}$ where
$B\Language(\TA,\strategy)$ is the set of prefixes of controlled traces ending with a time region symbol $\tick$ or~$\notick$.

Among those natural beliefs, we will be particularly interested in the ones showing \emph{leaks} of  information about the system. Intuitively, a \emph{leaking} natural belief allows to discriminate private and public runs.
For a given TA $\TA$, we denote
$\secret{\TA} = \big\{ \class{(\loc, \clockval)} \mid \loc \in \LocSetPriv, \clockval \in \setRgeqzero^{\ClockCard} \big\}$ the set of regions reachable after visiting $\locpriv$ on a run in~$\TA$, and
$\notsecret{\TA} = \big\{ \class{(\loc, \clockval)} \mid \loc \in \LocSetPub, \clockval \in \setRgeqzero^{\ClockCard} \big\}$ the set of regions reachable on a run not visiting $\locpriv$ in~$\TA$.\label{def:secret-duplication}

\begin{definition}\label{definition:leaking}
	Given a TA~$\TA$, a natural belief~$\belief{}$ is said to be \emph{leaking for full opacity} when exactly one of the following two conditions is satisfied:
	\begin{oneenumerate}%
		\item $(\belief{} \cap \finalclass{\TA} \cap \secret{\TA} \neq \emptyset)$, or
		\item $(\belief{} \cap \finalclass{\TA} \cap \notsecret{\TA} \neq \emptyset)$.
	\end{oneenumerate}%
\end{definition}
This means that finishing in this belief leaks an information to the attacker: only one final state is possible (private or public, but not both).

By definition of leaking natural beliefs and full opacity, we immediately have the following result:
\begin{lemma}\label{theorem:opacity-leaking-belief}
	Let $\TA$ be a TA and $\strategy$ a sequential strategy.
	$\TA$ is fully opaque with~$\strategy$ iff there is no belief in 
	$\beliefcontrolset{\TA}{\strategy}$ that is leaking for full opacity.
\end{lemma}

We will now build an automaton, called \emph{belief automaton}, constructing the natural beliefs from the observations and the choices of the sequential strategy.
In the following, $\Parts{\ControlSet}^{\leq N}$ denotes the set of sequences of length at most $N$ of subsets of $\ControlSet$.
Let $S = \ActionSet_1 \dots \ActionSet_n$ be a finite sequence of alphabets.
A path~$\run$ is \emph{locally $S$-compatible} if there exists a decomposition $\run = \run_1 \cdot (d_1, e_1) \cdot \run_2 \cdot (d_2, e_2) \dots (d_{n-1}, e_{n-1}) \cdot \run_n$ such that
$\ctr{\tw(\run_1)} \in (\ActionSet_1 \cup \set{u})^*$; and for $i \in \interval{1}{n-1}, \ctr{\tw((d_i, e_i) \cdot \run_{i+1})} \in (\ActionSet_{i+1} \cup \set{u})^*$ and $d_i > 0$.
Let $S = \ActionSet_1 \dots \ActionSet_N$ and $S' = \ActionSet_k \dots \ActionSet_N$ be two sequences in $\Parts{\ControlSet}^{\leq N}$, which we call \emph{local strategies}, suffix one of the other for a certain index~$i$.
Then we denote by $S \sqcap S'$ the local strategy $\ActionSet_1 \dots \ActionSet_k$.
Any path leading from a belief with local strategy~$S$ to another with local strategy~$S'$ must be locally $S \sqcap S'$-compatible.

Let $r$ be a region, $B$ be a set of regions, $a$ an action and $S, S'$ local strategies such that $S'$ is a suffix of~$S$. Then we define the following belief of~$r$ or of~$B$, reading $a$ and changing the strategy from $S$ to~$S'$:
\begin{align*}
\Next_a^{S'} ( r,S ) &= \big\{ r' \mid \exists r \in B, \exists p \text{ a path in } \TA \text{ with } \firststate(p) \in r, \laststate(p) \in r',\\
&\tw(p) = a \textit{ and } p \text{ is locally } (S \sqcap S') \text{-compatible.} \big\} \\
\Next_a^{S'} ( B,S ) &= \bigcup\limits_{r \in B} \Next_a^{S'} ( r,S )
\end{align*}

We extend the notation to a set~$H$ of pairs composed of a belief and a local strategy.

\[
\Succ_a ( H ) = \{ (B', S') \mid B' = \bigcup\limits_{(B, S) \in H \textit{ s.t.\ } S' \textit{ suffix of } S}  \Next_a^{S'} ( B,S ) \text{ with } a \in \set{\tick, \notick} \implies S' = \emptyset \}.\]

\begin{definition}[Belief Automaton]
	For a given TA~\mbox{$\TA = \TAcontrolextend$},
	the belief automaton~$\belaut{\TA}{}$ is given by the tuple
	$(\belset{\TA}, \belalpha, \belinit, \beltrans)$ where:
	\begin{itemize}%
		\item $\belset{\TA} = \Parts{\set{(\belief,{\localstrat}) \mid \belief \in \Parts{\regset{\TA}} \land \localstrat \in \Parts{\ControlSet}^{\leq N}}}\times \set{\controlflag, \nocontrolflag}$, hence every state is associated to a set of pairs composed of one belief and a local strategy of length at most $N$, as well as a tag in $\set{\controlflag, \nocontrolflag}$;
		\item $\belalpha = \Parts{\ControlSet}^{\leq N} \bigcup  \ObsSet$;
		\arXivVersion{\item }$\belinit = (\set{(\set{\regioni{0}}, \emptyset)} , \controlflag)$;
		\item $\beltrans = \beltrans_{\mathit{strat}} \cup \beltrans_o \subset \belset{\TA} \times \belalpha \times \belset{\TA}$ with
		\begin{enumerate}
			\item $\beltrans_{\mathit{strat}} = 
			\set{((\set{(\belief,{\emptyset})},\controlflag), \localstrat, (\set{(\belief,{\localstrat})},\nocontrolflag)) \mid \belief \in \Parts{\regset{\TA}}, \localstrat \in \Parts{\ControlSet}^{\leq N}}$,
			\item $\beltrans_o = \{ ((H,\nocontrolflag), a, (\Succ_{a}(H),z)) \mid H \in \belset{\TA}, a $ such that $\Succ_a(H) \neq \emptyset, z=\controlflag \text{ if } a\notin \set{\tick, \notick}, z=\nocontrolflag \text{ otherwise}\}$
		
		\end{enumerate}
	\end{itemize}%
\end{definition}

Intuitively, the belief automaton alternates between states that have a single belief 
and the tag $\controlflag$ 
where a local strategy
is selected, and states with the tag $\nocontrolflag$ associated to a set $H$ with many pairs of beliefs and sequences where the observable actions are selected.
A pair $(B,S)$ contained in a state of the belief automaton $H$ means that with the current observations, a run may be in the regions of $B$ with the remaining local strategy $S$ available for the runs. This represents the fact that throughout a time interval, we do 
not know the current implementation of the strategy, and thus when it moves to the next control alphabet.
As the local strategy has at most $N$ elements, a state $H$ can keep up to $N$ beliefs simultaneously depending on how far into the local strategy the runs gathered in this belief are.

\begin{definition}[Controlled Belief Automaton]
	For a given TA~\mbox{$\TA = \TAcontrolextend$} and a sequential strategy $\strategy$,
	the controlled belief automaton~$\belaut{\TA}{\strategy}$ is given by the tuple
	$(\belset{\TA}^{\strategy}, \belalpha, \belinit^{\strategy}, \beltrans^{\strategy})$ where:
	\begin{itemize}%
		\item $\belset{\TA}^{\strategy} = \belset{\TA} \times \ObsSet^* $;
		\arXivVersion{\item }$\belalpha = \Parts{\ControlSet}^{\leq N} \bigcup \ObsSet$;
		\arXivVersion{\item }$\belinit^{\strategy} = (\set{\regioni{0}}, \emptyset, \silentaction)$;
		\item $\beltrans^{\strategy} = \beltrans_{\mathit{strat}}^{\strategy} \cup \beltrans_o^{\strategy} \subset \belset{\TA}^{\strategy} \times \belalpha \times \belset{\TA}^{\strategy}$ with
		\begin{enumerate}
			\item $\beltrans_{\mathit{strat}}^{\strategy} = 
			\{((q, w), S , (q', w) \mid S = \strategy(w) \wedge (q,S,q')\in \beltrans_{\mathit{strat}}\}$;
			\item $\beltrans_o^{\strategy} = \{ (q, w), a, (q', w\cdot a) \mid (q, a, q') \in \beltrans_o \}$

		\end{enumerate}
	\end{itemize}%
\end{definition}

We now aim to prove that the natural beliefs correspond to the beliefs of the form $( \{(B, \emptyset)\}, \controlflag, w )$, called \emph{encountered beliefs}, by showing that these beliefs are reachable if and only if for all region $r$ in $B$, there exists a run $\strategy$-compatible which trace is $w$, starting in the region $r_0$ and ending in $r$.

\begin{restatable}{lemma}{naturaloutcome}\label{lem:naturaloutcome}
The natural beliefs are exactly the encountered beliefs in the controlled  belief automaton.
\end{restatable}

The proof of this lemma is in \cref{appendix:proof:lem:naturaloutcome}.
Thanks to this result, the \ProblemNSS{} problem can be translated into a safety
two-player game on the belief automaton where the first player selects the strategy and aims to avoid the leaking encountered belief, while the second player selects the
buffered observation. The translation of the \ProblemNSSRC{} problem is slightly more difficult, as we need to combine the safety condition and the reachability condition into
a single Büchi condition for the game.
In both cases, we translate the opacity problems into a doubly exponential game, that can be solved in polynomial time in the size of the game---hence the \TwoEXPTIME{} algorithm.
The details of this game as well as the hardness proof are also in \cref{appendix:proof:NSS}.

\section{Ensuring opacity with Observable Sequential Strategies (OSS)}\label{sec:OSS}

In a control setting, it is natural to assume that controllable events are observable, because the controller can only choose or enable events it can detect: if it could not observe them, it would not know when or whether its control decisions are being applied.
This is the motivation behind the OSS setting.
An unfortunate consequence of this assumption however is that allowing more controllable actions does not help make the system opaque, and thus there exists a strategy ensuring weak opacity iff the strategy $\strategyBlocksAll$ that constantly blocks every controllable action ensures weak opacity.
The proofs of this section are in \cref{sec:annexOSS}.

\begin{restatable}{proposition}{OSSprop}
\label{cor:eptystrat}
Given a TA $\TA$ in the OSS setting, the existence of a sequential strategy making $\TA$ weakly opaque implies that $\strategyBlocksAll$ makes $\TA$ weakly opaque.
\end{restatable}

The \EXPSPACE-completeness of the opacity problem without control hence immediately gives:

\begin{restatable}{theorem}{OSS}\label{theorem-OSS}
The \ProblemOSS{} problem is \EXPSPACE-complete.
\end{restatable}

Requiring the sequential strategy to be \reachcondition{} avoids---at least partially---the previous situation, as blocking every controllable action may not allow any run to reach a final location.
Nonetheless, we can similarly limit the shape of the \reachcondition{} sequential strategies that achieve opacity: such a strategy only needs to allow \emph{a single run}
leading to a final location, and any deviation from that run only has to satisfy opacity and thus can be
supervised by the strategy~$\strategyBlocksAll$.
Moreover, we can produce a doubly exponential bound for the length of the path leading to the final location.
This bound can then be used to limit the sequences of control alphabet the sequential strategy needs to use, hence reducing the \ProblemOSSRC{} problem to the \ProblemNSS{} problem.

\begin{restatable}{theorem}{OSSRC}\label{theorem-OSS-NB}
The \ProblemOSSRC{} problem is in \ThreeEXPTIME and \TwoEXPTIME-hard.
\end{restatable}

\section{Conclusion}

\begin{table}[tb]
	\centering
	\caption{Summary: impact of the attacker observation power on the decidability of the existence of a sequential strategy ensuring opacity}
	\label{table-summary-power}
	\setlength{\tabcolsep}{2pt} %
	\begin{tabular}{l c }
		\hline
		\rowHeader{}\cellHeader{Attacker observation power} & \cellHeader{Decidability} \\
		\hline
		\cellHeader{}All observations & \cellUndecidableExisting{}undecidable \cite{Cassez09}\\
		\hline
		\cellHeader{}Execution time only & \cellDecidableExisting{}In \EXPSPACE{} \cite{ADLL26}\\
		\hline
		\cellHeader{}Buffered observations & \cellUndecidable{}undecidable (\cref{theorem:general-undecidability})\\
		\hline
	\end{tabular}

\end{table}

\begin{table}[tb]
	\centering
	\caption{Summary: decidability of the existence of a strategy ensuring opacity under buffered observations}
	\label{table-summary-strategies}
	\setlength{\tabcolsep}{2pt} %
	\scalebox{.9}{
	\begin{tabular}{l c | c }
		\hline
		\rowHeader{}\cellHeader{Type of strategy} & \cellHeader{General strategy}  & \cellHeader{Non-blocking strategy} \\
		\hline
		\cellHeader{}Sequential & \multicolumn{2}{c}{\cellUndecidable{}undecidable (\cref{theorem:general-undecidability})}\\
		\hline
		\cellHeader{}$N$-sequential & \multicolumn{2}{c}{\cellDecidable{}\TwoEXPTIME{}-complete (\cref{theorem:NSS})}\\
		\hline
		\cellHeader{}Observable sequential & \cellDecidable{}\EXPSPACE{}-complete (\cref{theorem-OSS}) & \cellDecidable{}\ThreeEXPTIME /\TwoEXPTIME-hard (\cref{theorem-OSS-NB})\\
		\hline
	\end{tabular}
	}

\end{table}

We studied the problem of controlling timed automata to ensure opacity under \emph{buffered observations}, a realistic attacker model in which the attacker observes all actions within each time unit, in the right order but without their precise timestamps.
We first proved that full and weak opacity are inter-reducible, and that they are both undecidable in this controlled setting.
We then identified two independent and practically motivated restrictions of the controller, both restoring decidability alone:
when the strategy changes at most $N$~times per time unit ($N$-sequential strategies),
and when all controllable actions are observable and distinguishable by the attacker (observable sequential strategies), the problems become decidable, and in most cases \TwoEXPTIME-complete (with the exception of the \ProblemOSS{} problem which is \EXPSPACE-complete).
We summarise our results in \cref{table-summary-power,table-summary-strategies}.

Even though we identified realistic decidable classes of TAs, the theoretical complexity of solving opacity for these models is high and, due to our hardness results, it cannot be substantially improved.
In order to be able to handle those problems in practice, one needs either to find a relevant subclass of TAs with a lower complexity, or to design heuristics that---despite high theoretical complexities---remain practically efficient.
The former could consist for instance in \emph{bounding the observer's memory} across time units.
The latter would require a completely different approach from ours, as our techniques strongly rely on the region automaton, which is known to be practically expensive.
Moreover, the \emph{zone automaton} (a usual structure to circumvent the region automaton) does not pair well with the determinisation step we need to construct the belief automaton.
One could for instance adapt the notion of control over buffered observation in an \emph{untimed} setting, hence removing entirely the need for the region automaton.
In another direction, we will investigate \emph{parametric variants}, where $N$ or some elements of the TA are not part of the input, hence allowing to represent systems that are not fully defined.

	\newcommand{\CCIS}{Communications in Computer and Information Science}
	\newcommand{\CSUR}{{ACM} Computing Surveys}
	\newcommand{\ENTCS}{Electronic Notes in Theoretical Computer Science}
	\newcommand{\FAC}{Formal Aspects of Computing}
	\newcommand{\FundInf}{Fundamenta Informaticae}
	\newcommand{\FMSD}{Formal Methods in System Design}
	\newcommand{\IJFCS}{International Journal of Foundations of Computer Science}
	\newcommand{\IJSSE}{International Journal of Secure Software Engineering}
	\newcommand{\IJC}{International Journal of Control}
	\newcommand{\IseCure}{International Journal of Information Security}
	\newcommand{\IPL}{Information Processing Letters}
	\newcommand{\IC}{Information and Computation}
	\newcommand{\JAIR}{Journal of Artificial Intelligence Research}
	\newcommand{\JLAP}{Journal of Logic and Algebraic Programming}
	\newcommand{\JLAMP}{Journal of Logical and Algebraic Methods in Programming} %
	\newcommand{\JISA}{Journal of Information Security and Applications}
	\newcommand{\JLC}{Journal of Logic and Computation}
	\newcommand{\JALC}{Journal of Automata, Languages and Combinatorics}
	\newcommand{\LMCS}{Logical Methods in Computer Science}
	\newcommand{\LNCS}{Lecture Notes in Computer Science}
	\newcommand{\RESS}{Reliability Engineering \& System Safety}
	\newcommand{\RTS}{Real-Time Systems}
	\newcommand{\SCP}{Science of Computer Programming}
	\newcommand{\SOSYM}{Software and Systems Modeling ({SoSyM})}
	\newcommand{\STTT}{International Journal on Software Tools for Technology Transfer}
	\newcommand{\TCS}{Theoretical Computer Science}
	\newcommand{\TOPLAS}{{ACM} Transactions on Programming Languages and Systems ({ToPLAS})}
	\newcommand{\ToPNoC}{Transactions on {P}etri Nets and Other Models of Concurrency}
	\newcommand{\TOSEM}{{ACM} Transactions on Software Engineering and Methodology ({ToSEM})}
	\newcommand{\TSE}{{IEEE} Transactions on Software Engineering}
	\newcommand{\TCAD}{{IEEE} Transactions on Computer-Aided Design of Integrated Circuits and Systems}

\ifdefined\VersionAuthor
	\renewcommand*{\bibfont}{\small}
	\printbibliography[title={References}]
\else
	\newpage
	\bibliographystyle{plainurl}%
	\bibliography{PTA}

\begin{thebibliography}{10}

\bibitem{AD94}
Rajeev Alur and David~L. Dill.
\newblock A theory of timed automata.
\newblock {\em \TCS{}}, 126(2):183--235, April 1994.
\newblock \href {https://doi.org/10.1016/0304-3975(94)90010-8}
  {\path{doi:10.1016/0304-3975(94)90010-8}}.

\bibitem{AFH99}
Rajeev Alur, Limor Fix, and Thomas~A. Henzinger.
\newblock Event-clock automata: {A} determinizable class of timed automata.
\newblock {\em \TCS{}}, 211(1-2):253--273, 1999.
\newblock \href {https://doi.org/10.1016/S0304-3975(97)00173-4}
  {\path{doi:10.1016/S0304-3975(97)00173-4}}.

\bibitem{AEYM21}
Ikhlass Ammar, Yamen El~Touati, Moez Yeddes, and John Mullins.
\newblock Bounded opacity for timed systems.
\newblock {\em \JISA{}}, 61:1--13, September 2021.
\newblock \href {https://doi.org/10.1016/j.jisa.2021.102926}
  {\path{doi:10.1016/j.jisa.2021.102926}}.

\bibitem{AALL26}
Sjood {Ammen Daje}, Tareq {Ahmad Al-Sarayrah}, Ding Liu, and Zhiwu Li.
\newblock Robust opacity in timed automata: notions and verification.
\newblock {\em Reliability Engineering \& System Safety}, 270:112093, 2026.
\newblock URL:
  \url{https://www.sciencedirect.com/science/article/pii/S095183202501292X},
  \href {https://doi.org/10.1016/j.ress.2025.112093}
  {\path{doi:10.1016/j.ress.2025.112093}}.

\bibitem{AGWZH24}
Jie An, Qiang Gao, Lingtai Wang, Naijun Zhan, and Ichiro Hasuo.
\newblock The opacity of timed automata.
\newblock In André Platzer, Kristin-Yvonne Rozier, Matteo Pradella, and Matteo
  Rossi, editors, {\em {FM}}, volume 14933 of {\em \LNCS{}}, pages 620--637.
  Springer, 2024.
\newblock \href {https://doi.org/10.1007/978-3-031-71162-6_32}
  {\path{doi:10.1007/978-3-031-71162-6_32}}.

\bibitem{AAL24}
{\'E}tienne Andr{\'e}, Johan Arcile, and Engel Lefaucheux.
\newblock Execution-time opacity problems in one-clock parametric timed
  automata.
\newblock In Siddharth Barman and Sławomir Lasota, editors, {\em {FSTTCS}},
  volume 323 of {\em Leibniz International Proceedings in Informatics
  (LIPIcs)}, pages 3:1--3:22. Schloss Dagstuhl -- Leibniz-Zentrum für
  Informatik, December 2024.
\newblock \href {https://doi.org/10.4230/LIPIcs.FSTTCS.2024.3}
  {\path{doi:10.4230/LIPIcs.FSTTCS.2024.3}}.

\bibitem{ABLM22}
{\'E}tienne Andr{\'e}, Shapagat Bolat, Engel Lefaucheux, and Dylan Marinho.
\newblock {strategFTO}: Untimed control for timed opacity.
\newblock In Cyrille Artho and Peter Ölveczky, editors, {\em {FTSCS}}, pages
  27--33. {ACM}, 2022.
\newblock \href {https://doi.org/10.1145/3563822.3568013}
  {\path{doi:10.1145/3563822.3568013}}.

\bibitem{ADLL26}
{\'E}tienne Andr{\'e}, Marie Duflot, Laetitia Laversa, and Engel Lefaucheux.
\newblock Execution-time opacity control for timed automata.
\newblock {\em Software and Systems Modeling}, 2026.
\newblock To appear.
\newblock \href {https://doi.org/10.1007/s10270-026-01395-5}
  {\path{doi:10.1007/s10270-026-01395-5}}.

\bibitem{ADL26}
{\'E}tienne Andr{\'e}, Sarah Dépernet, and Engel Lefaucheux.
\newblock The bright side of timed opacity.
\newblock {\em \LMCS{}}, 2026.
\newblock To appear.
\newblock URL: \url{https://arxiv.org/pdf/2408.12240v5}.

\bibitem{AK20}
{\'E}tienne Andr{\'e} and Aleksander Kryukov.
\newblock Parametric non-interference in timed automata.
\newblock In Yi~Li and Alan Liew, editors, {\em {ICECCS}}, pages 37--42, 2020.
\newblock \href {https://doi.org/10.1109/ICECCS51672.2020.00012}
  {\path{doi:10.1109/ICECCS51672.2020.00012}}.

\bibitem{ALLMS23}
{\'E}tienne Andr{\'e}, Engel Lefaucheux, Didier Lime, Dylan Marinho, and Jun
  Sun.
\newblock Configuring timing parameters to ensure execution-time opacity in
  timed automata.
\newblock In Maurice~H. ter Beek and Clemens Dubslaff, editors, {\em {TiCSA}},
  volume 392 of {\em Electronic Proceedings in Theoretical Computer Science},
  pages 1--26, 2023.
\newblock Invited paper.
\newblock \href {https://doi.org/10.4204/EPTCS.392.1}
  {\path{doi:10.4204/EPTCS.392.1}}.

\bibitem{ALM23}
{\'E}tienne Andr{\'e}, Engel Lefaucheux, and Dylan Marinho.
\newblock Expiring opacity problems in parametric timed automata.
\newblock In Yamine Ait-Ameur and Ferhat Khendek, editors, {\em {ICECCS}},
  pages 89--98, 2023.
\newblock \href {https://doi.org/10.1109/ICECCS59891.2023.00020}
  {\path{doi:10.1109/ICECCS59891.2023.00020}}.

\bibitem{AA23survey}
Johan Arcile and {\'E}tienne Andr{\'e}.
\newblock Timed automata as a formalism for expressing security: A survey on
  theory and practice.
\newblock {\em \CSUR{}}, 55(6):1--36, July 2023.
\newblock \href {https://doi.org/10.1145/3534967} {\path{doi:10.1145/3534967}}.

\bibitem{BDST02}
Roberto Barbuti, Nicoletta~De Francesco, Antonella Santone, and Luca Tesei.
\newblock A notion of non-interference for timed automata.
\newblock {\em \FundInf{}}, 51(1-2):1--11, 2002.

\bibitem{BT03}
Roberto Barbuti and Luca Tesei.
\newblock A decidable notion of timed non-interference.
\newblock {\em \FundInf{}}, 54(2-3):137--150, 2003.

\bibitem{BCLR15}
Gilles Benattar, Franck Cassez, Didier Lime, and Olivier~H. Roux.
\newblock Control and synthesis of non-interferent timed systems.
\newblock {\em \IJC{}}, 88(2):217--236, 2015.
\newblock \href {https://doi.org/10.1080/00207179.2014.944356}
  {\path{doi:10.1080/00207179.2014.944356}}.

\bibitem{BFHHH14}
Nathalie Bertrand, Eric Fabre, Stefan Haar, Serge Haddad, and Loïc Hélouët.
\newblock Active diagnosis for probabilistic systems.
\newblock In Anca Muscholl, editor, {\em {FoSSaCS}}, volume 8412 of {\em
  \LNCS{}}, pages 29--42. Springer, 2014.
\newblock \href {https://doi.org/10.1007/978-3-642-54830-7_2}
  {\path{doi:10.1007/978-3-642-54830-7_2}}.

\bibitem{Cassez09}
Franck Cassez.
\newblock The dark side of timed opacity.
\newblock In Jong~Hyuk Park, Hsiao{-}Hwa Chen, Mohammed Atiquzzaman, Changhoon
  Lee, Tai{-}Hoon Kim, and Sang{-}Soo Yeo, editors, {\em {ISA}}, volume 5576 of
  {\em \LNCS{}}, pages 21--30. Springer, 2009.
\newblock \href {https://doi.org/10.1007/978-3-642-02617-1_3}
  {\path{doi:10.1007/978-3-642-02617-1_3}}.

\bibitem{CH12}
Krishnendu Chatterjee and Monika Henzinger.
\newblock An {$O(n^2)$} time algorithm for alternating {B}üchi games.
\newblock In Yuval Rabani, editor, {\em {SODA}}, pages 1386--1399. {SIAM},
  2012.
\newblock \href {https://doi.org/10.1137/1.9781611973099.109}
  {\path{doi:10.1137/1.9781611973099.109}}.

\bibitem{DQYang25}
Weilin Deng, Daowen Qiu, and Jingkai Yang.
\newblock Initial-location opacity and infinite-step opacity of timed automata
  with integer resets.
\newblock {\em IEEE Control Systems Letters}, 9:2031--2036, 2025.
\newblock \href {https://doi.org/10.1109/LCSYS.2025.3590422}
  {\path{doi:10.1109/LCSYS.2025.3590422}}.

\bibitem{DQY25}
Weilin Deng, Daowen Qiu, and Jingkai Yang.
\newblock New insights into opacity verification in timed discrete-event
  systems.
\newblock {\em Automatica}, 186:112869, 2026.
\newblock \href {https://doi.org/10.1016/J.AUTOMATICA.2026.112869}
  {\path{doi:10.1016/J.AUTOMATICA.2026.112869}}.

\bibitem{Dima01}
Catalin Dima.
\newblock Real-time automata.
\newblock {\em \JALC{}}, 6(1):3--23, 2001.
\newblock \href {https://doi.org/10.25596/jalc-2001-003}
  {\path{doi:10.25596/jalc-2001-003}}.

\bibitem{DLL25}
Yifan Dong, Dimitri Lefebvre, and Zhiwu Li.
\newblock $k$-step opacity verification and enforcement of time labeled {P}etri
  net systems.
\newblock {\em {IEEE} Transactions on Automatic Control}, 70(9):5848--5863,
  2025.
\newblock \href {https://doi.org/10.1109/TAC.2025.3552020}
  {\path{doi:10.1109/TAC.2025.3552020}}.

\bibitem{KKG24}
Julian Klein, Paul Kogel, and Sabine Glesner.
\newblock Verifying opacity of discrete-timed automata.
\newblock In Nico Plat, Stefania Gnesi, Carlo~A. Furia, and Antónia Lopes,
  editors, {\em {FormaliSE}}, pages 55--65. {ACM}, 2024.
\newblock \href {https://doi.org/10.1145/3644033.3644376}
  {\path{doi:10.1145/3644033.3644376}}.

\bibitem{LLHL22}
Jun Li, Dimitri Lefebvre, Christoforos~N. Hadjicostis, and Zhiwu Li.
\newblock Observers for a class of timed automata based on elapsed time graphs.
\newblock {\em {IEEE} Transactions on Automatic Control}, 67(2):767--779, 2022.
\newblock \href {https://doi.org/10.1109/TAC.2021.3064542}
  {\path{doi:10.1109/TAC.2021.3064542}}.

\bibitem{LLHL25}
Jun Li, Dimitri Lefebvre, Christoforos~N. Hadjicostis, and Zhiwu Li.
\newblock Verification of state-based timed opacity for constant-time labeled
  automata.
\newblock {\em {IEEE} Transactions on Automatic Control}, 70(1):503--509, 2025.
\newblock \href {https://doi.org/10.1109/TAC.2024.3432788}
  {\path{doi:10.1109/TAC.2024.3432788}}.

\bibitem{PF12}
Hans-Jörg Peter and Bernd Finkbeiner.
\newblock The complexity of bounded synthesis for timed control with partial
  observability.
\newblock In Marcin Jurdzinski and Dejan Nickovic, editors, {\em {FORMATS}},
  volume 7595 of {\em \LNCS{}}, pages 204--219. Springer, 2012.
\newblock \href {https://doi.org/10.1007/978-3-642-33365-1_15}
  {\path{doi:10.1007/978-3-642-33365-1_15}}.

\bibitem{Rintanen04}
Jussi Rintanen.
\newblock Complexity of planning with partial observability.
\newblock In Shlomo Zilberstein, Jana Koehler, and Sven Koenig, editors, {\em
  {ICAPS}}, pages 345--354. {AAAI}, 2004.

\bibitem{SLR23}
Anthony Spriet, Didier Lime, and Olivier~H. Roux.
\newblock Timed non-interference under partial observability and bounded
  memory.
\newblock In Laure Petrucci and Jeremy Sproston, editors, {\em {FORMATS}},
  volume 14138 of {\em \LNCS{}}, pages 122--137. Springer, 2023.
\newblock \href {https://doi.org/10.1007/978-3-031-42626-1_8}
  {\path{doi:10.1007/978-3-031-42626-1_8}}.

\bibitem{WZ18}
Lingtai Wang and Naijun Zhan.
\newblock Decidability of the initial-state opacity of real-time automata.
\newblock In Cliff~B. Jones, Ji~Wang, and Naijun Zhan, editors, {\em Symposium
  on Real-Time and Hybrid Systems - Essays Dedicated to Professor Chaochen Zhou
  on the Occasion of His 80th Birthday}, volume 11180 of {\em \LNCS{}}, pages
  44--60. Springer, 2018.
\newblock \href {https://doi.org/10.1007/978-3-030-01461-2_3}
  {\path{doi:10.1007/978-3-030-01461-2_3}}.

\bibitem{WZA18}
Lingtai Wang, Naijun Zhan, and Jie An.
\newblock The opacity of real-time automata.
\newblock {\em \TCAD{}}, 37(11):2845--2856, 2018.
\newblock \href {https://doi.org/10.1109/TCAD.2018.2857363}
  {\path{doi:10.1109/TCAD.2018.2857363}}.

\bibitem{Zhang24}
Kuize Zhang.
\newblock State-based opacity of labeled real-time automata.
\newblock {\em \TCS{}}, 987:114373, 2024.
\newblock \href {https://doi.org/10.1016/J.TCS.2023.114373}
  {\path{doi:10.1016/J.TCS.2023.114373}}.

\end{thebibliography}
\fi

\newpage
\appendix

\section{Recalling the region automaton}\label{appendix:regions}

We recall the classic region automaton construction.
Given a TA~$\TA$, for a clock~$\clock_i$, we denote by~$\constantmax{i}$ the largest constant to which~$\clock_i$ is compared within the guards and invariants of~$\TA$:
 formally:
\CONCURVersion{%
	\(\constantmax{i} = \max_j\{ \paramd_j \mid  x_i \bowtie \paramd_j \text{ appears in a guard or invariant of }\TA \}\).
}%
\arXivVersion{%
	\[\constantmax{i} = \max_j\{ \paramd_j \mid  x_i \bowtie \paramd_j \text{ appears in a guard or invariant of }\TA\text{.} \}\]
}

Given a clock valuation $\clockval$ and a clock~$\clock_i$,
$\intpart{\clockval(\clock_i)}$ (resp.\ $\fract{\clockval(\clock_i)}$) denotes the integral (resp.\ fractional) part of $\clockval(\clock_i)$.
\arXivVersion{%

}%
We now recall the equivalence relation between clock valuations.
\begin{definition}[Equivalence relation~\cite{AD94}]
	Two clocks valuations $\clockval,\clockval'$ are \emph{equivalent}, denoted by $\clockval \clockeq \clockval'$, when the following three conditions hold for any clocks $\clock_i, \clock_j \in \ClockSet$:
	\begin{enumerate}
		\item $\intpart{\clockval(\clocki{i})} = \intpart{\clockval'(\clocki{i})}$ or $\clockval(\clocki{i}) > \constantmax{i}$ and $\clockval'(\clocki{i}) > \constantmax{i}$; \label{item:equiv-int}
		\item if  $\clockval(\clocki{i}) \leq \constantmax{i}$ and $\clockval(\clocki{j}) \leq \constantmax{j}$: $\fract{\clockval(\clock_i)} \leq \fract{\clockval(\clock_j)}$ iff $\fract{\clockval'(\clock_i)} \leq \fract{\clockval'(\clock_j)}$;\arXivVersion{ and}\label{item:equiv-order}
		\item if  $\clockval(\clocki{i}) \leq \constantmax{i}$: $\fract{\clockval(\clock_i)} = 0$ iff $\fract{\clockval'(\clock_i)} = 0$.\label{item:equiv-zero}
	\end{enumerate}
\end{definition}
\arXivVersion{%
In other words, two valuations are equivalent when, for a given clock,
	the integral part is the same in both valuations or greater than the largest constant in both valuations (\cref{item:equiv-int}),
	for any two clocks, the fractional parts are in the same order in the two valuations (\cref{item:equiv-order}),
	and
	the fractional part is zero in both valuations or neither (\cref{item:equiv-zero}).

The equivalence relation $\clockeq$ is extended to the states of $\semantics{\TA}$:
given
}\CONCURVersion{Given}
two states $\TTSstate = (\loc, \clockval), \TTSstate' = (\loc', \clockval')$ of $\semantics{\TA}$, we write $\TTSstate \clockeq \TTSstate'$ iff $\loc = \loc'$ and $\clockval \clockeq \clockval'$.
We denote by $\class{\TTSstate}$ and call \emph{region} the equivalence class of a state~$\TTSstate$ for~$\clockeq$.
Then, $\TTSstate' \in \class{\TTSstate}$ when $\TTSstate \clockeq \TTSstate'$.
The set of all regions of~$\TA$ is denoted~$\regset{\TA}$.
A region $\region = \class{(\loc, \clockval)}$ is \emph{final} whenever $\loc \in \LocFinalSet$.
The set of final regions is denoted by~$\finalclass{\TA}$.
A region~$\region$ is \emph{reachable} when there exists a run~$\run$ such that $\laststate(\run) \in \region$ ($\laststate(\run)$ denoting the final configuration of $\run$.

Given a state $\TTSstate = (\loc, \clockval)$, and $\paramd \in \setRgeqzero$, we write $\TTSstate + \paramd $ to denote $(\loc, \clockval + \paramd)$.
Given two regions $r$ and $r'$, we write $\region \cup \region'$ for $ \set{\TTSstate \mid \TTSstate \in \region \mbox{ or } \TTSstate \in \region'}$.

\begin{definition}[Region Automaton]
	For a given TA~$\TA$,
	the region automaton~$\regaut{\TA}$ is given by the tuple
	$(\regset{\TA}, \ActionControlSet \cup \set{\silentaction}, \RegAutTransitions, r_0 =(l_0,\ClocksZero))$ where:
	\begin{oneenumerate}%
		\item $\regset{\TA}$ is the set of states,
		\item  given two regions $\region, \region' \in \regset{\TA}$ and $a \in \ActionControlSet \cup \set{\silentaction}$, we have $(\region,\action,\region') \in \RegAutTransitions$ if there exist $s =(\loc, \clockval) \in \region$ and $ s'=(\loc', \clockval') \in \region'$ such that one of the following holds:
		\begin{ienumerate}%
			\item $\action \in \ActionControlSet$ and
			$(\loc, \clockval) \transitionWith{\edge} (\loc', \clockval') \in \TTStransition $ in $\TTS_{\TA}$
			with $\edge = (\loc, \guard, \action, \resets, \loc')$
			for some $\guard$ and~$\resets$; or
		\item $\action = \silentaction$ and  $\exists \paramd \in \setRpositif$ such that
$\TTSstate \transitionWith{\paramd} \TTSstate'$,
and $\forall  0< \paramd' < \paramd$, $\TTSstate+d' \in r \cup r'$.
		\end{ienumerate}
	\end{oneenumerate}%
\end{definition}
\section{Buffered automaton}\label{sec:buffaut}

Given a TA~$\TA$, we give the construction of a TA~$\buffered{\TA}$ whose untimed language is exactly the set of buffered observations of~$\TA$.
Simply said, we make two copies of the TA, encoding the Boolean value whether the elapsed time is an integer or not, and add a clock~$z$ (encoding the fractional part of the global time) to switch between both parts at the right moment.

Each location~$\loc$ is duplicated in two locations $\loc^0$ and~$\loc^1$.
A transition $(\loc, a, g, R, \loc')$ in~$\TA$ is decomposed into two new transitions $(\loc^i, a, g \land g_i, R, \loc'^i)$ in $\buffered{\TA}$, for $i \in \set{0, 1}$, with $g_0 = (z = 0)$ and $g_1 = (0 < z) \land (z< 1)$.

To switch between both sets of locations, we add for each location~$\loc$ in~$\TA$ two transitions in $\buffered{\TA}$: $(\loc^1, \clabel{\tick}{\uncontrol}, (z = 1), \set{0}, \loc^0)$ (denoting that the global time has reached an integer value), and $(\loc^0, \clabel{\notick}{\uncontrol}, g_1, \emptyset, \loc^1)$ (denoting that some time smaller than~1 time unit has elapsed since the last integer global time).
This way, the untimed trace of a timed word contains time marks allowing to know between which integer times an action took place but also whether an action took place exactly at an integer time.

If $\loc$ (in~$\TA$) was an initial location, then $\loc^0$ is initial; however if $\loc$ was final, then none of $\loc^0$ and~$\loc^1$ are final, for the following reason.

In the following proofs of the appendix, we restrict the opacity study to states reached by a transition labelled by a time region symbol ($\tick$ or $\notick$). Hence all relevant events of a run must happen before the last time region symbol occurs, which we add artificially as explained below.
To end with a time region symbol without forgetting when the run ended, we first add four locations $\loc_{\textit{end}}^0$, $\loc_{\textit{end}}^1$, $\loc_{\textit{reg}}^0$ and $\loc_{\textit{reg}}^1$, and a special observable symbol $\runend$ that marks the end of the run. If $\loc$ was a final location, then for $i \in \set{0,1}$, $\loc^i$ leads uncontrollably and in zero time to $\loc_{\textit{end}}^i$, producing the symbol $\runend$. To ensure that this is done in zero time, every transition leading to $\loc_{\textit{end}}^i$ resets a clock $\clocky$, then compared to zero. Finally, we add two new final locations $\loc_{\textit{reg}}^0$ and $\loc_{\textit{reg}}^1$ respectively reachable from $\loc_{\textit{end}}^1$ and $\loc_{\textit{end}}^0$ through a transition announcing the next time region (with $\clabel{\tick}{\uncontrol}$ or $\clabel{\notick}{\uncontrol}$ resp.).

With this construction, the set of durations of runs in~$\TA$ is exactly the set of times when a symbol $\runend$ can be produced in $\buffered{\TA}$.

\begin{figure}[tb]

	\centering
	\large
	\scalebox{.9}{%
	\begin{tikzpicture}[pta, scale=2, xscale=1, yscale=.5]
		\location[initial, name=s0, at={(0, -1)}]{$\styleclock{\clock} \leq 3$}{$\loci{0}^0$}

		\location[private, name=s2, at={(2, 0.5)}]{$\styleclock{\clock} \leq 2$}{$\locpriv^0$}

		\node[location] at (4, -1) (s1) {$\locfinal^0$};
		
		\node[location] at (5, 0.5) (s3) {$\loc_{\textit{end}}^0$};
		\node[location, final] at (6, -1) (s4) {$\loc_{\textit{reg}}^0$};

		\path (s0) edge[bend left] node[align=center, above]{$\styleclock{\clock} \geq 1$ \\ $\land \tickclock = 0$\\$\clabel{\silentaction}{\cactioni{1}}$} (s2);
		\path (s0) edge[loop above] node[align=center]{$\clabel{\action}{\cactioni{2}}$\\$\tickclock = 0$} (s0);
		\path (s0) edge[] node[above right, align=center]{$\clabel{\actionb}{\uncontrol}$\\$\tickclock = 0$\\$\styleclock{\clocky} \assign 0$} (s1);
		\path (s2) edge[bend left] node[align=center]{$\clabel{\actionb}{\uncontrol}$\\$\tickclock = 0$\\$\styleclock{\clocky} \assign 0$} (s1);
		\path (s1) edge node[above left, align=center]{$\clabel{\runend}{\uncontrol}$\\$\styleclock{\clocky} = 0$} (s3);

		\location[name=s0', at={(0, -3.5)}]{$\styleclock{\clock} \leq 3$}{$\loci{0}^1$}

		\location[private, name=s2', at={(2, -5)}]{$\styleclock{\clock} \leq 2$}{$\locpriv^1$}

		\node[location] at (4, -3.5) (s1') {$\locfinal^1$};
		
		\node[location] at (5, -5) (s3') {$\loc_{\textit{end}}^1$};
		\node[location, final] at (5, -3.5) (s4') {$\loc_{\textit{reg}}^1$};

		\path (s0') edge[bend right] node[align=center, below]{$\clabel{\silentaction}{\cactioni{1}}$ \\ $\styleclock{\clock} \geq 1$ \\ $\land 0 < \tickclock < 1$} (s2');
		\path (s0') edge[loop left] node[align=center, below]{$\clabel{\action}{\cactioni{2}}$\\ $0 < \tickclock < 1$} (s0');
		\path (s0') edge[] node[below right, align=center]{$\clabel{\actionb}{\uncontrol}$\\$0 < \tickclock < 1$\\$\styleclock{\clocky} \assign 0$} (s1');
		\path (s2') edge[bend right] node[align=center, below]{$\clabel{\actionb}{\uncontrol}$\\$0 < \tickclock < 1$\\$\styleclock{\clocky} \assign 0$} (s1');
		\path (s1') edge node[below left, align=center]{$\clabel{\runend}{\uncontrol}$\\$\styleclock{\clocky} = 0$} (s3');

		\path (s0) edge[bend right=+9] node[align=center, left] {$\clabel{\notick}{\uncontrol}$\\$0 < \tickclock < 1$} (s0');
		\path (s2) edge[bend right=+9] node[align=center, left] {$\clabel{\notick}{\uncontrol}$\\$0 < \tickclock < 1$} (s2');
		\path (s0') edge[bend right=+9] node[align=center, right] {$\tick$\\$\tickclock = 1$\\$\tickclock \assign 0$} (s0);
		\path (s2') edge[bend right=+9] node[align=center, right] {$\tick$\\$\tickclock = 1$\\$\tickclock \assign 0$} (s2);
		\path (s3) edge node[left, align=center]{$\clabel{\notick}{\uncontrol}$\\$0 < \tickclock < 1$} (s4');
		\path (s3') edge node[below right, align=center]{$\tick$\\$\tickclock = 1$\\$\tickclock \assign 0$} (s4);

	\draw[color=black, -, thick, decorate,decoration={brace,raise=0.1cm}]
(6.25,1.75) --++ (0,-3.5) node[right=0.2cm,pos=0.5, align=center] {Integer \\ global time};
	\draw[color=black, -, thick, decorate,decoration={brace,raise=0.1cm}]
(6.25,-2.75) --++ (0,-3.5) node[right=0.2cm,pos=0.5, align=center] {Non-integer \\ global time};
	\end{tikzpicture}
	}
	\caption{The buffered automaton of the TA from \cref{figure:example-TA}}
	\label{figure:buffered-TA}

\end{figure}
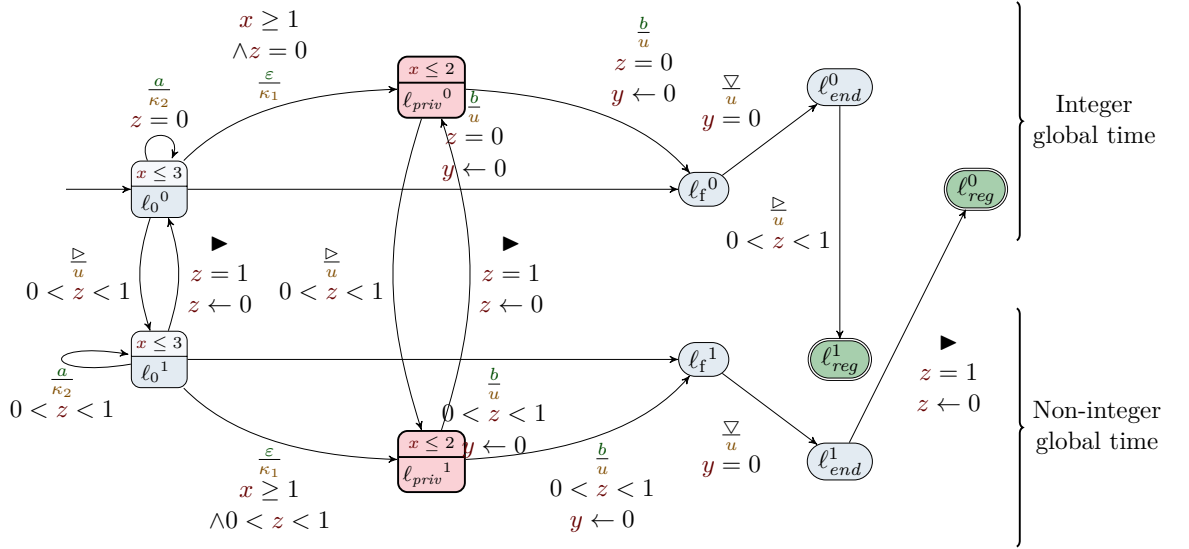

\begin{remark}
Each run in~$\TA$ has one equivalent run, or infinitely many ones, in~$\buffered{\TA}$, and each run in $\buffered{\TA}$ corresponds to exactly one run in~$\TA$.
\end{remark}
For instance, the timed word $(a_1, 0.2)(a_2, 0.7)(a_3, 1.4)(a_4, 2)(a_5, 3.6)$ in~$\TA$ becomes the set of timed words
\begin{align*}
\{&(\notick, \tau_0)(a_1, 0.2)(a_2, 0.7)(\tick, 1)(\notick, \tau_1)(a_3, 1.4)(\tick, 2)(a_4, 2)(\notick, \tau_2)(\tick, 3)(\notick, \tau_3)(a_5, 3.6)
\\ &\mid 
0 < \tau_0 \leq 0.2 \land 0 < \tau_1 \leq 0.4 \land 0 < \tau_2 < 1 \land 0 < \tau_3 \leq 0.6\}.
\end{align*}
They have a same untimed trace: $\notick a_1 a_2 \tick \notick a_3 \tick a_4 \notick \tick \notick a_5$ which gives an incomplete and succinct timing information about the timestamps $t_i$ of each $a_i$ that can be understood as: $0 < t_1 \leq t_2 < 1 < t_3 < 2 = t_4 < 3 < t_5$.

\section{Inter-reduction between weak and full opacity}
\label{sec:interreduc}

\interreduc*

\begin{proof}
In~\cite{ADL26}, the weak ET-opacity and full ET-opacity problems were shown to be inter-reducible in the absence of control.
The proof mainly relied on the ability to construct from a
TA~$\TA$ another TA~$\TA'$ where private and public runs were swapped.
Hence, $\TA$ is fully opaque iff both $\TA$ and~$\TA'$ are weakly opaque.
This proof cannot carry over to the controlled setting since there is no guarantee that the same strategy can be applied in both TAs to make them \emph{simultaneously} weakly opaque:
both $\TA$ and $\TA'$ could be deemed weakly ET-opaque, though relying on two different sequential strategies.
While their proof cannot be adapted, we reuse some of their proof ingredients,
mainly $\APriv$ and~$\APub$, \ie{} copies of the TA~$\TA$ encoding respectively all private and public behaviours (see details in \cref{section:tools:Apriv-Apub}):
given a sequential strategy $\strategy$, we have that 
$\PrivateTr{\TA}{\strategy} = \Trace{\CRuns{\strategy}{\APriv}}$ and
$\PublicTr{\TA}{\strategy} = \Trace{\CRuns{\strategy}{\APub}}$.

Let us first explain how to reduce weak opacity problems to full opacity problems.
Given a TA $\TA$, we first build the TA $\APriv$ and~$\APub$ mentioned above, and then
the TA $\TAWeakFull$ (shown on the left of \cref{fig:weak-full}) 
such that the initial location of $\APriv$ and~$\APub$ can be reached in $0$ time 
while visiting a private location, and the initial location of $\APub$ can be visited in $0$ time without visiting a private location.
Thus, $\PrivateTr{\TAWeakFull}{\strategy} =  \Trace{\CRuns{\strategy}{\APriv}}\cup 
\Trace{\CRuns{\strategy}{\APub}} = \Trace{\CRuns{\strategy}{\TA}}$ and
$\PublicTr{\TAWeakFull}{\strategy}= \Trace{\CRuns{\strategy}{\APub}} = \PublicTr{\TA}{\strategy}$. 
Hence, 
$\PrivateTr{\TA}{\strategy} \subseteq \PublicTr{\TA}{\strategy}$ if and only if $\PrivateTr{\TAWeakFull}{\strategy} = \PublicTr{\TAWeakFull}{\strategy}$.
Thus the strategy~$\strategy$ makes $\TA$ weakly opaque if and only if it makes 
$\TAWeakFull$ fully opaque.

We now consider the opposite reduction. As mentioned previously, in~\cite{ADL26} we had made this reduction by testing both inclusion
$\PrivateTr{\TA}{\strategy} \subseteq \PublicTr{\TA}{\strategy}$
and $\PrivateTr{\TA}{\strategy} \supseteq \PublicTr{\TA}{\strategy}$ separately
on two different TAs. Here, instead, we use a single TA, and add a final letter (`$a$' or `$b$') pointing
which inclusion is being tested: if the last letter of the trace is an `$a$', then the run participates in the test of
the inclusion $\PrivateTr{\TA}{\strategy} \subseteq \PublicTr{\TA}{\strategy}$; and if the last letter is a `$b$', then the run participates in the converse inclusion.
As a strategy is not aware before this last letter of which inclusion is being tested, it needs to ensure both simultaneously.

More precisely, we detail below the construction of the TA $\TAFullWeak$ (shown on the right of \cref{fig:weak-full}), where we assume `$\extraAction$' is a new letter (and thus cannot be produced through any transition besides the ones we add in the reduction):

\begin{itemize}
\item From a new initial location $\locinit$, one can go to the initial location of
either $\APriv$ or~$\APub$ in $0$ time;
\item Within $\APriv$ and~$\APub$, each transition leading to a 
final location of $\APriv$ or~$\APub$ resets $x$;
and these final locations are not deemed final in $\TAFullWeak$;
\item Once a final locations of $\APriv$ is reached, a run can immediately (requiring $x=0$) leave while reading  a $\extraAction$ and go either to $\loc^a$ (which is private) or to
$\loc^b$.
Similarly, once a final locations of $\APub$ is reached, a run can immediately (requiring $x=0$) leave while reading a $\extraAction$ and go either to $\loc_a$ or to
$\loc_b$ (which is private).

\item From the locations $\loc^z$ and $\loc_z$ (with $z\in \{a,b\}$) the run can finally reach the final location $\locfinal$ of $\TAFullWeak$ by reading a $z$.
\end{itemize}

\begin{figure}[th]
\begin{center}
\begin{tikzpicture}[pta, node distance=2cm, thin]
\def\shift{-4.5} %
	\node[location, initial] (L0) at (0 +\shift, 0) {$\locinit$};
	\node[location, private] (Lpriv) at (1.35+\shift, 1) {$\locpriv$};
	\node (A) at (3.25+\shift, 1) {};
	\node[rectangle, minimum width=1.5cm, minimum height=1cm, align=center, draw] (rectangleA) at (3.7+\shift, 1cm) {$\APriv$};

	\node (B) at (3.25+\shift, -1) {};
	\node[rectangle, minimum width=1.5cm, minimum height=1cm, align=center, draw] (rectangleB) at (3.7+\shift, -1cm) {$\APub$};
	
	\path
	(L0) edge node[pos=0.5, above left, align=center]{$\silentaction$\\$\styleclock{\clock} = 0$}(Lpriv)
	(L0) edge node[pos=0.5, below left, align=center]{$\silentaction$\\$\styleclock{\clock} = 0$}(B)
	(Lpriv) edge node[pos=0.5, above, align=center]{$\silentaction$\\$\styleclock{\clock} = 0$} (A)
	(Lpriv) edge node[pos=0.1, below, align=center]{$\silentaction$\\$\styleclock{\clock} = 0$} (B)
	;

\def\vampl{1.25} %

	\node[location, initial] (l0) at (1, 0) {$\locinit$};
\node[location, final] (lf) at (+8, 0) {$\locfinal$};
	\node[location, private] (lpriv1) at (6.5, 1.5*\vampl) {$\loc^a$};
	\node[location, private] (lpriv2) at (6.5, -1.5*\vampl) {$\loc_b$};
	\node[location] (lpub1) at (6.5, 0.5*\vampl) {$\loc^b$};
	\node[location] (lpub2) at (6.5, -0.5*\vampl) {$\loc_a$};

	\node (ta1) at (1.5, 1*\vampl) {};
	\node[location] (1at) at (4.35, 1*\vampl) {};
	\node[rectangle, minimum width=3.5cm, minimum height=\vampl cm, align=center, draw] (rectangle) at (3cm, \vampl cm) {$\APriv \hspace{50pt}$};

	\node (ta2) at (1.5, -1*\vampl) {};
	\node[location] (2at) at (4.35, -1*\vampl) {};
	\node[rectangle, minimum width=3.5cm, minimum height=\vampl cm, align=center, draw] (rectangle) at (3cm, -1*\vampl cm) {$\APub \hspace{45pt}$};
	
	\path
	(1at)+(-1.3,0) edge[dashed] node{$\styleclock{\clock} \assign 0$}(1at)
	(2at)+(-1.3,0) edge[dashed] node{$\styleclock{\clock} \assign 0$}(2at)
	(l0) edge node[pos=0.5, above left, align=center]{$\silentaction$\\$\styleclock{\clock} = 0$} (ta1)
	(l0) edgenode[pos=0.5, below left, align=center]{$\silentaction$\\$\styleclock{\clock} = 0$} (ta2)
	(1at) edge node[pos=0.5, above, align=center]{$\extraAction$\\$\styleclock{\clock} = 0$} (lpriv1)
	(1at) edge node[pos=0.3, below, align=center]{$\extraAction$\\$\styleclock{\clock} = 0$} (lpub1)
	(2at) edge node[pos=0.5, above, align=center]{$\extraAction$\\$\styleclock{\clock} = 0$} (lpub2)
	(2at) edge node[pos=0.3, below, align=center]{$\extraAction$\\$\styleclock{\clock} = 0$} (lpriv2)
	(lpriv1) edge node[pos=0.5, above right]{$a$}(lf)
	(lpub1) edge node[pos=0.4, above]{$b$}(lf)
	(lpub2) edge node[pos=0.4, below]{$a$}(lf)
	(lpriv2) edge node[pos=0.5, below right]{$b$}(lf)
	;
	
\end{tikzpicture}
\end{center}
  \caption{Left: $\TAWeakFull$, fully opaque iff $\TA$ is weakly opaque; right: $\TAFullWeak$ weakly opaque iff $\TA$ is fully opaque. Every transition outside $\APriv$ and~$\APub$ is uncontrollable.}
  \label{fig:weak-full}
\end{figure}
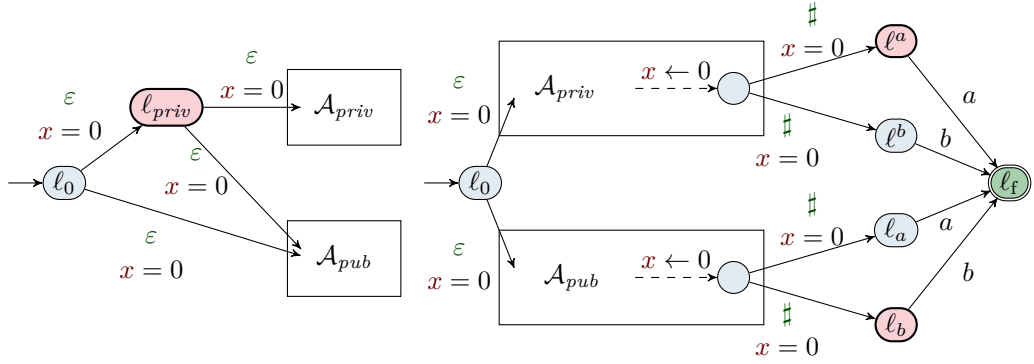

By construction, we have that a trace $w$ is produced by a private run of 
$\TA$ (and thus by a run of $\APriv$) iff the trace $w\extraAction a$ is a private trace 
of $\TAFullWeak$ and the trace $w\extraAction b$ is a public trace 
of $\TAFullWeak$. And also that a trace $w$ is produced by a public run of 
$\TA$ (and thus by a run of $\APub$) iff the trace $w\extraAction a$ is a public trace 
of $\TAFullWeak$ and the trace $w\extraAction b$ is a private trace 
of $\TAFullWeak$.

As $\extraAction$ can only be produced once by the runs exiting the TA $\APriv$ and 
$\APub$,
given a sequential strategy $\strategy$, in order to have 
$\PrivateTr{\TAFullWeak}{\strategy} \subseteq \PublicTr{\TAFullWeak}{\strategy}$
we need that the traces of runs exiting $\APriv$ (resp.~$\APub$) and followed by $\extraAction a$ (resp.~$\extraAction b$) are made opaque by traces of runs exiting $\APub$ (resp.~$\APriv$)  and followed by $\extraAction a$ (resp.~$\extraAction b$).
Those two conditions thus require that 
$\Trace{\CRuns{\strategy}{\APriv}}\subseteq \Trace{\CRuns{\strategy}{\APub}}$
and 
$\Trace{\CRuns{\strategy}{\APriv}}\supseteq \Trace{\CRuns{\strategy}{\APub}}$
respectively, and thus that
$\Trace{\CRuns{\strategy}{\APriv}}= \Trace{\CRuns{\strategy}{\APub}}$.

Hence, 
$\PrivateTr{\TAFullWeak}{\strategy} \subseteq \PublicTr{\TAFullWeak}{\strategy}$
iff $\PrivateTr{\TA}{\strategy} = \PublicTr{\TA}{\strategy}$.

As the $N$-SS setting only limits the strategy to consider, and we do not add any controllable action in our reduction, this inter-reducibility directly applies to both the $N$-SS setting and the OSS setting as well as the general one.
\end{proof}

\section{$\APriv$ and $\APub$}\label{section:tools:Apriv-Apub}

In this section, we give a quick intuition of the constructions $\APub$ and~$\APriv$, that recognize timed words produced respectively by public and private runs of a given TA~$\TA$.
Those constructions are also illustrated in \cref{fig:two-automata} and
\cref{figure:APriv-APub:example-TA}. We refer to~\cite{ADL26} for the formal definitions.

The public runs TA $\APub$ is the easiest to build: it suffices to remove the private locations from~$\TA$. The remaining runs are exactly the ones that do not visit the private locations, thus the public runs.

The private runs TA~$\APriv$ is obtained by duplicating all locations and transitions of~$\TA$: one copy~$\TA_{S}$ corresponds to the paths that already visited the private locations set, and the other copy $\TA_{\bar{S}}$ corresponds to the paths that did not (this is a usual way to encode a Boolean, here ``$\PrivSet$ was visited'', in the locations of a~TA).
For each private location $\locpriv$ in~$\TA$, we redirect all transitions leading to the copy of $\locpriv$ in~$\TA_{\bar{S}}$ towards the copy of~$\locpriv$ in~$\TA_{S}$.
The initial location is the one from~$\TA_{\bar{S}}$ and the final locations are the ones from~$\TA_{S}$.
Hence all runs need to go from $\TA_{\bar{S}}$ to $\TA_{S}$ before reaching a final location, which requires visiting a private location.

\begin{figure}[htb]
	\centering
	\begin{subfigure}[hb]{0.4\textwidth}
   \centering
\begin{tikzpicture}[pta, node distance=2cm, thin]

	\node[location, initial] (l0) at (-1, 1) {$\locinit$};
	\node[location, final] (lfS) at (+1, 1) {$\locfinal$};
	\node[rectangle, minimum width=3cm, minimum height=2cm, align=center, draw] (rectangle) at (0cm, 1cm) {$\TA$};

	\node[location, private] (lpriv) at (0, 0.5) {$\locpriv$};
	\path
	;

\draw[-, very thick, color = black!20] (-0.5,0.2) -- (0.5,0.8);
\draw[-, very thick, color = black!20] (0.5,0.2) -- (-0.5,0.8);

\end{tikzpicture}
  \caption{$\APub$}
  \label{figure:APub}
	\end{subfigure}
	\begin{subfigure}[hb]{0.4\textwidth}
   \centering
	\begin{tikzpicture}[pta, node distance=2cm, thin]

	\node[location] (l0S) at (-1, 0) {$\locinit^S$};
	\node[location, final] (lfS) at (+1, 0) {$\locfinal^S$};
	\node[rectangle, minimum width=3cm, minimum height=1.5cm, align=center, draw] (rectangle) at (0cm, -.1cm) {$\TA_S$\\};

	\node[location, initial] (l0S) at (-1, -2) {$\locinit^{\bar{S}}$};
	\node[location] (lfS) at (+1, -2) {$\locfinal^{\bar{S}}$};
	\node[rectangle, minimum width=3cm, minimum height=1.5cm, align=center, draw] (rectangle) at (0cm, -2cm) {$\TA_{\bar{S}}$};

	\node[location, private] (lpriv) at (0, -0.5) {$\locpriv^S$};
	\path
	(-0.5, -1.5) edge (lpriv)
	;
	\end{tikzpicture}
  \caption{$\APriv$}
  \label{figure:APriv}
	\end{subfigure}
	\caption{Illustrating $\APub$ and $\APriv$ 
	}
	\label{fig:two-automata}
\end{figure}
\begin{figure}[htb]

	\centering
	\begin{subfigure}[hb]{0.3\textwidth}
		\begin{tikzpicture}[pta, xscale=1, yscale=.5]

		\location[initial, name=s0, at={(0, -1)}]{$\styleclock{\clock} \leq 3$}{$\loci{0}$}

		\node[location, final] at (1.5, -1) (s1) {$\loci{1}$};

		\path (s0) edge[loop above] node[align=center]{$\styleclock{\clock} = 1$\\$\cactioni{2}$} (s0);
		\path (s0) edge[] node[below]{$\uncontrol$} (s1);

		\end{tikzpicture}
		\caption{$\APub$}
		\label{fig:example-APub}
	\end{subfigure}
	\hfill
	\begin{subfigure}[hb]{0.65\textwidth}
		\begin{tikzpicture}[pta, scale=2, xscale=1, yscale=.45]
		\node[rectangle, minimum width=6.25cm, minimum height=2.2cm, align=center, draw, dashed] (rectangleS) at (1.25cm, 2.25cm) {};
		\location[name=s0S, at={(0, 1.75)}]{$\styleclock{\clock} \leq 3$}{$\loci{0}$};

		\location[private, name=s2S, at={(1.5, 2.75)}]{$\styleclock{\clock} \leq 2$}{$\loci{2}$};

		\node[location, final] at (2.5, 1.75) (s1S) {$\loci{1}$};

		\node[rectangle, minimum width=6.25cm, minimum height=2.2cm, align=center, draw, dashed] (rectangleSbar) at (1.25cm, -0.5cm) {};
		\location[initial, name=s0, at={(0, -1)}]{$\styleclock{\clock} \leq 3$}{$\loci{0}$}

		\location[name=s2, at={(1.5, 0)}]{$\styleclock{\clock} \leq 2$}{$\loci{2}$}

		\node[location] at (2.5, -1) (s1) {$\loci{1}$};

		\path (s0S) edge node[align=center, pos=0.6]{$\styleclock{\clock} \geq 1$\\$\cactioni{1}$} (s2S);
		\draw (s0S) .. controls +(100:20pt) and +(130:15pt) .. (s0S) node[pos=0.5, align=center,above]{$\styleclock{\clock} = 1$\\$\cactioni{2}$};
		\path (s0S) edge[] node[below]{$\uncontrol$} (s1S);
		\path (s2S) edge node[]{$\uncontrol$} (s1S);
		\draw (s0) .. controls +(100:20pt) and +(130:15pt) .. (s0) node[pos=0.5, align=center,above]{$\styleclock{\clock} = 1$\\$\cactioni{2}$};
		\path (s0) edge[] node[below]{$\uncontrol$} (s1);
		\path (s2) edge node[]{$\uncontrol$} (s1);
		\path (s0) edge node[align=center, below right, pos=0.4]{$\styleclock{\clock} \geq 1$\\$\cactioni{1}$} (s2S);
		\end{tikzpicture}
		\caption{$\APriv$}
		\label{fig:example-APriv}
	\end{subfigure}
\caption{$\APub$ and $\APriv$ with the example from \cref{figure:example-TA}}
\label{figure:APriv-APub:example-TA}
\end{figure}

\section{Proof for $N$-sequential strategies}\label{appendix:proof:NSS}
\subsection{Proof of \cref{lem:naturaloutcome}}\label{appendix:proof:lem:naturaloutcome}

\naturaloutcome*

\textbf{Proof ingredients}
\begin{itemize}
\item[(a)]\label{item:succ-set}
The set $\Succ_a(H)$ is by definition the greatest set such that:
$\forall (B,S) \in H, \forall \region \in B, \forall (\loc, \clockval) \in \region, \forall r' \in \regset{\TA}, \forall (\loc', \clockval' \in r', \forall S' \textit{ suffix of }S, (\exists \textit{ a locally } S~\sqcap~S' \textit{-compatible path } p \textit{ from } (\loc, \clockval) \textit{ to }(\loc', \clockval')  \implies \exists B' \textit{ s.t.\ } (B', S') \in \Succ_a(H) \land r' \in B'$.
\item[(b)]\label{item:strategy-glue}
Let $p$ and $p'$ be two paths in~$\TA$ such that $\laststate(p) = \firststate(p')$, and let $S_1$ be a local strategy. Let $S_2$ be a suffix of $S_1$ and $S_3$ a suffix of $S_2$ such that $p$ is locally $S_1 \sqcap S_2$-compatible and $p'$ is locally $S_2 \sqcap S_3$-compatible. Then the path $p \cdot p'$ is locally $S_1 \sqcap S_3$-compatible.
\end{itemize}

\begin{lemma}\label{lemma:succ-back}
Let $a$ be some letter in $\ObsSet$ and $S$ be some local strategy. Let $r$ and~$r'$ be regions such that there exist some configurations $(\loc, \clockval) \in r$ and $(\loc', \clockval') \in r'$ and some locally $S$-compatible path $\run$ in~$\TA$ from $(\loc, \clockval)$ to $(\loc', \clockval')$, whose trace is~$a$. Then for all configuration $(\loc', \clockval'')$ in~$r'$, there exists a locally $S$-compatible path in~$\TA$, with trace~$a$, from some configuration in~$r$ to $(\loc', \clockval'')$.
\end{lemma}

\begin{proof}
This proof reuses some of the ideas of \cite[Proposition~6.7]{ADL26}%
, where more technical details can be found.

Let $(\loc', \clockval'')$ be some configuration in $r'$. We denote by $t_0$ the global time at which the path $\run$ begins.
Since $\run$ is locally $S$-compatible, there exists an implementation enabling all its transitions. We denote by $\nu$ this implementation.
From $\run$ and $\nu$ we derive the increasing sequence of timestamps~$f$ of all transitions taken including the timestamps of~$\nu$.
The valuation $\tilde{\clockval}$ of each clock $\clock$ that is not reset during $\run$ is defined by subtracting the duration of $\run$ from its valuation $\clockval'(\clock)$.
For each reset of some clocks occurring in~$\run$ at some time~$t$, we can shift the corresponding transition's timestamps to fit the last clock valuation $\clockval''$, getting a new timestamp~$t'$. Repeating this process for each transition resetting clocks gives us the timestamps sequence $(t'_i)_i$ from~$(t_i)_i$. The clock valuations $\clockval'$ and $\clockval''$ share for each clock the same integral parts, and for each pair of clocks, the same (strict) ordering of their fractional parts. Thus the new sequence of timestamps $(t'_i)_i$ and the former one $(t_i)_i$ also share the same integral parts and ordering of the fractional parts. We add as last element of both these sequences the time $t_0 + \duration(\run)$.

We define the distortion function from $(t_i)_i$ to~$(t'_i)_i$ as follows:
\[\begin{array}{rrcl}
\gamma_{t \rightarrow t'}: & [t_0;t_0 + \duration(\run)) & \longrightarrow & [t_0; t_0 + \duration(\run)) \\
 & t & \longmapsto & t'_j + \frac{t'_{j+1}-t'_j}{t_{j+1}-t_j} (t-t_j) \text{ where $t_j,t_{j+1}$ are such that } t_j \leq t < t_{j+1}\text{.} \\
\end{array}\]

This function is increasing so the sequence $(\gamma_{t \rightarrow t'} (f_i) )_i$ shares the same integral parts and orderings of fractional parts as $(f_i)_i$. Modifying $\run$ and $\nu$ with this new sequence of timestamps allows us to build a path $\run'$ in~$\TA$ and an implementation $\nu'$ of $S$ such that all transitions of $\run'$ are enabled by $\nu'$, $\run'$ has trace $a$, $\firststate(\run') = (\loc, \tilde{\clockval})$ and $\laststate(\run') = (\loc', \clockval'')$. Indeed, every positive delay in $\run$ with $\nu$ is also positive in $\run'$ with $\nu'$, and conversely.

\end{proof}

We now prove both inclusions of \cref{lem:naturaloutcome} separately.
\begin{proof}[Proof of \cref{lem:naturaloutcome}]
$(\Rightarrow)$ We show here the inclusion of the set of encountered beliefs in $\beliefcontrolset{\TA}{\strategy}$,\ie{} we show that each path in the controlled belief automaton corresponds to at least one $\strategy$-compatible path in the TA.
Let $P$ be some path $(H_0, \silentaction, \textit{cflag}_0) \rightarrow (H_1, w_1, \textit{cflag}_1) \rightarrow \dots \rightarrow (H_n, w_1 w_2 \dots w_n, \textit{cflag}_n)$ in $\belaut{\TA}{\strategy}$%
.

We show by induction that there is a $\strategy$-compatible run in~$\TA$ with trace $w = w_1 w_2 \dots w_n$, going from the initial configuration in $H_0$ to some configuration $(\loc', \clockval')$ in some region of a belief $B_n$ of $H_n$.%

\textbf{Induction step:} For some $1 \leq k \leq n$, let $(\loc_k, \clockval_k)$ be an arbitrarily chosen configuration in some region of a belief $B_k$ in $H_k$. Thanks to the definition of $\Succ_{w_k}(H_{k-1})$, combined with \cref{lemma:succ-back}, there are some belief $(B_{k-1}, S_{k-1})$ in $H_{k-1}$, some region $\region_{k-1} \in B_{k-1}$ and some configuration $(\loci{k-1}, \clockval_{k-1})$ such that a locally $S_{k-1}\sqcap S_{k}$-compatible path starting in $(\loci{k-1}, \clockval_{k-1})$ and ending in $(\loc_k, \clockval_k)$ with trace $w_k$ exists in~$\TA$. We denote this path $\run_k$.

We backward iterate this step from the end of the path to its beginning to progressively define all paths parts $\run_k$ with their connecting configurations $(\loc_k, \clockval_k)$. The only transitions from $P$ that are not translated into some transitions in~$\TA$ are those simulating the controller's choices. They are simply ignored, since they let the belief set unchanged.

Once the initial belief set $H_0$ is reached, since the only configuration in some belief of $H_0$ is the initial configuration $(\loc_0, \clockval_0)$, the path $\run_1$ starts in $(\loc_0, \clockval_0)$. Hence, the resulting path $\run = \run_1 \cdot \run_2 \dots \run_n$ is a run if it ends in some final configuration. Moreover, from proof ingredient (b), we can show that $\run_n$ is locally $S_n$-compatible (since $S_n \sqcap \emptyset = S_n$), $\run_{n-1} \cdot \run_n$ is locally $S_{n-1}$-compatible, etc. until showing that $\run$ is locally $\strategy(w)$-compatible, \ie{} it is $\strategy$-compatible.

$(\Leftarrow)$ We now address the converse inclusion, showing that every path in~$\TA$ is represented by a path in $\belaut{\TA}{\strategy}$.

Let $\run$ be some $\strategy$-compatible path starting in the initial configuration in~$\TA$. We write it:
$\run = (\loc_0,\clockval_0) \longuefleche{(d_1, \edge_1)} (\loc_1,\clockval_1) \longuefleche{(d_2, \edge_2)} \dots \longuefleche{(d_n, \edge_n)} (\loc_n,\clockval_n) $ for some $n$, and denote by $w$ its trace. We build the corresponding path in $\belaut{\TA}{\strategy}$.

Since $\run$ is $\strategy$-compatible, there exists an implementation $\nu$ of $\strategy$ enabling all transitions of $\run$. Then for each time region $\set{k}$ or $(k;k+1)$, $\nu$ associates each alphabet change $\ActionSet_{k,i} \ActionSet_{k,i+1}$ with a time $t_{k,i+1}$. For each observable action $w_i$ in $w$, we denote by $j(i)$ the index of the corresponding observable transition in $\run$. 
This gives a decomposition of $\run$ into $\run_1 \cdot \run_2 \dots \run_{\vert w \vert}$ where each part of path $\run_i$ ends with the $j(i)$-th transition, of observable label $w_i$ (and thus, cannot extend over several time regions). For $1 \leq i \leq \vert w \vert$, we denote by $r_i$ the region of $\firststate(\run^i)$. 
The implementation $\nu$ also associates each path $\run_i$, ending in the time region $\set{k}$ or $(k;k+1)$, with the longest suffix $S_{k, i} = \ActionSet_{k,m} \ActionSet_{k,m+1} \dots \ActionSet_{k,N}$ of $\strategy(\Trace{\run_1 \cdot \dots \cdot \run_i})$ such that $t_{k, m} \leq \duration(\run_1 \dots \run_i)$. 
As the definitions of local compatibility and compatibility are strongly related, for all $0 \leq i \leq \vert v \vert$ such that $w_i \notin \set{\tick, \notick}$, the path $\run_i$ is locally $S_{k,i} \sqcap S_{k,i+1}$-compatible for some $k$. If $w_i \in \set{\tick, \notick}$, $\run_i$ is the last path before a change of time region, and it is $S_{k,i}$-compatible for some $k$.
Unfolding the definitions of $\Next_{w_i}^{S'}(r,S)$,$\Next_{w_i}^{S'}(B,S)$ and $\Succ_{w_i}(H)$, we obtain:
$\Succ_{w_i} ( H ) = \{ (B', S') \mid B' = \bigcup\limits_{(B, S) \in H \textit{ s.t.\ } S' \textit{ suffix of } S}
 \bigcup\limits_{r \in B} \{ r' \mid \exists r \in B, \exists p \text{ a path in } \TA \text{ with } \firststate(p) \in r, \laststate(p) \in r', \buffered{\tw(p)} = w_i \textit{ and } p \text{ is locally } (S \sqcap S') \text{-compatible.}\}
 \text{ with } w_i \in \set{\tick, \notick} \implies S' = \emptyset \}$
 and thus, since the path $\run_1$ meets the requirements in this definition, it defines the transition between the initial set of beliefs $(H_0, \silentaction, \controlflag)$ and its successor $(H_1, w_1, \textit{cflag}_1)$ with $H_1 = \Succ_{w_1} ( H_0 )$. Since $r_i \in H_i = \Succ_{w_i} ( H_{i-1}$, the other paths $p_i$ also define the following transitions, finally reaching $H_n$, the set of beliefs where ends the corresponding path in $\belaut{\TA}{\strategy}$.

\end{proof}

\subsection{Turning the belief automaton into a game}\label{sec:gameNSS}

From \cref{theorem:opacity-leaking-belief} and \cref{lem:naturaloutcome}, a sequential strategy ensures full opacity iff none of the encountered beliefs in its controlled belief automaton are leaking.
Intuitively, finding such a strategy amounts to solving a two-player B\"uchi game on the belief automaton, with the first player selecting the choice of the sequential strategy
(the transitions $\beltrans_{\mathit{strat}}^{\strategy}$), while the second select the observation that the TA will produce with this choice (the transitions 
$\beltrans_o^{\strategy}$), the goal of the first player being to avoid leaking encountered beliefs.

More formally, a two-player safety game can be defined by a tuple 
$\mathcal{G}=((Q_1,Q_2), q_0,\delta,\Sigma, \mathcal{T})$ where $Q_1$ (resp.~$Q_2$) are the states controlled by the first (resp.~second) player, $q_0$ is the initial state,
$\Sigma$ is an alphabet of actions,
$\delta\subseteq (Q_1\cup Q_2)\times\Sigma\times(Q_1\cup Q_2)$ describes the available transitions, 
and $\mathcal{T}\subseteq (Q_1\cup Q_2)$ is the set of states that the first player targets.
Starting from $q_0$, the player controlling the current state selects one of the transition exiting this state to reach a new state. The first player wins iff
it visits a state of $\mathcal{T}$ infinitely often.

\begin{lemma}[\cite{CH12}]\label{lem:buchi-ptime}
Deciding the existence of a winning strategy for the first player in a two-player B\"uchi
game can be done in \PTIME.
\end{lemma}

Combining this result with the doubly exponential size of the belief automaton produces 
the algorithmic part of

\buffopaque*

\begin{proof}
We start with \ProblemNSS{}.
Given a TA $\TA$, we extend the notion of leaking belief to the belief automaton: the states of the belief automaton of the form $(\{(B,\emptyset\}, \controlflag)$ are called leaking if $B$ is leaking.
Using the belief automaton of $\TA$ as a basis, we build the two player 
safety game $\mathcal{G}=((Q_1,Q_2),q_0,\delta,\belalpha, \mathcal{T})$ where
\begin{itemize}
\item $Q_1$ (resp.~$Q_2$) is the set of states of $\belset{\TA}$ tagged by $\controlflag$ (resp.~$\nocontrolflag$);
\item $q_0 = \belinit $;
\item $(q,z,q')\in \delta$ iff $(q,z,q')\in \beltrans$ and $q$ is not leaking,
\item $\mathcal{T}$ is the set of state of $Q_1\cup Q_2$ that are not leaking.
\end{itemize}

Reaching a leaking state incurs an immediate loss for the first player as
no transition exit those states. Moreover, any run that does not reach a leaking state 
is winning for the first player as the run is infinite and every state visited by the run is in $\mathcal{T}$.
Hence, player~1 has a winning strategy in $\mathcal{G}$ if and only if there exists a sequential strategy ensuring full opacity of $\TA$.

By \cref{lem:buchi-ptime}, deciding the winner of $\mathcal{G}$ takes \PTIME in
$|\mathcal{G}|$. Moreover, $G$ is doubly exponential in the input.
Indeed, denoting $H$ the number of clocks in~$\TA$ and $M$ the
largest constant appearing in any guard or invariant,
the number of regions of $A$ is $|\regaut{\TA}| \leq (2M+2)^H \cdot H! \cdot |\LocSet|$.
A belief is a set of regions of $\regaut{\TA}$, so the number of distinct beliefs is at most $2^{|R_A|}$, which is doubly exponential in the description of $\TA$.
A state of the belief automaton $\belaut{\TA}{}$ further pairs each belief $B \subseteq \regaut{\TA}$ with a sequence $S$ of controllable alphabets. The number of such sequences is at most $\left(2^{|\ControlSet|}\right)^N$, which is singly exponential in~$\TA$,
and a state of $\belaut{\TA}{}$ may contain up to $N$ different pairs of beliefs and 
sequence. There thus exists at most $(2^{|\ControlSet|})^N \cdot (2^{|R_A|})^N$ such 
states.
Hence $|\belaut{\TA}{}|$ is doubly exponential in the input size.

In conclusion, the overall procedure runs in \TwoEXPTIME.

For \ProblemNSSRC{}, we need to slightly tweak the construction of the game to include the 
two objectives.
The states are enhanced with a boolean value $r$ starting at $0$.
As long as $r=0$, player 1 controls every choice of transition within the game, with the exception that player 2 can decide to switch $r$ to $1$ after every selection of 
sequence $S$ (after any transition taken in $\beltrans_{\mathit{strat}}$).
As long as $r=0$, the first player wins if it can reach a state with the $\controlflag$ 
tag that contains a final location and is not a leaking state (note that we keep that
no transition exits leaking states, hence player 1 must ensure no leaks occur on the way).
If $r$ is switched to $1$, the game acts exactly as for \ProblemNSS{}, and player 1 must only ensure opacity.

If there exists a \reachcondition{} strategy ensuring opacity, player 1 plays this strategy 
over the selection of transitions in $\beltrans_{\mathit{strat}}$, and selects the observations corresponding to one of the existing runs that are compatible with the strategy for the transition in $\beltrans_o$ as long as $r=0$. As the strategy 
ensures opacity, no leaking belief will be visited, and if player 2 does not modify $r$, 
a state containing a final location will be reached as the strategy follows a run that
reaches a final location.

If there exists a strategy to win this game, one can define the sequential strategy
that imitates the choices in the game assuming $r=0$ as long as the sequence of observations is the one planned by the strategy of player 1, and imitates the choices assuming $r=1$ as soon as the observation deviates.
This strategy is \reachcondition, as there exists a run producing the 
sequence of observations that keeps $r=0$ in the game, and which is thus compatible with the strategy. It ensures opacity as playing as if $r=1$ ensures opacity following the
proof for \ProblemNSS{}, and following $r=0$ requires visiting no leaking belief as well
to avoid being stuck and thus ensures opacity.

As this construction only doubles the size of the game, the algorithm remains in \TwoEXPTIME.
\end{proof}

\subsection{Hardness}\label{subsection:NSShard}

In this section, we prove the \TwoEXPTIME hardness of the problem of existence of an $N$-sequential strategy making the TA opaque.
This proof is inspired by the reduction from~\cite{PF12}, itself based on a proof by Rintanen \cite{Rintanen04}.
The authors reduce the halting problem of an alternating Turing machine with exponential memory (known to be \TwoEXPTIME-complete) to the existence of a discrete controller (given as a discrete-time TA) such that the parallel execution of the given TA and of the controller always avoids a set of bad states. We adapt this proof to reduce the halting problem of an alternating Turing machine to the existence of a $1$-SS strategy making the TA fully opaque.

The reduction consists in defining a TA~$\TA$ which, controlled with a 1-sequential strategy, simulates the behaviour of the machine. We thus suppose \wlogen{} that the alternating Turing machine of exponential memory is given under the form $\machine = (\StatesExists, \StatesForall, q_0, q_f, \machinealph, \delta)$ with $\StatesExists \uplus \StatesForall$ the set of states, $q_0$ and $q_f$ the initial and final states, $\machinealph$ a two-symbol alphabet, and the transition function $\delta: Q \times \machinealph \rightarrow Q \times Q \times \machinealph \times \set{L, R}$ allowing to reach at most two states at each step. We suppose its tape has length $n = 2^b$. For the reduction to be polynomial, it is necessary to store the content of the tape not in the TA but with the strategy, using a possibly doubly exponentially long time.

If the strategy does not comply with the correct simulation of $\machine$, a ``bad'' location becomes reachable which uncontrollably leads to the production of non-opaque traces toward a sink state.

Let $t$ be a transition of~$\machine$, reading the symbol~$m$ on the tape and leading from a state~$q$ to the pair of states $(q', q'')$, reading $m$ and writing $m'$ on the tape; and moving the tape head in the direction~$D$.
We denote by~$\cur$ the current position of the tape head among $\interval{0}{n-1}$ (initially $\cur = 0$).
In the rest of this proof, we sometimes omit the control or observable label when they are identical.
The correct simulation of the transition~$t$ consists in several aspects, simultaneously managed in an interval of $2n$ time units:
\begin{enumerate}
\item
\begin{itemize}
\item The clock $\clock$ is equal to $n - \cur$ and $\clocky = 0$. The TA and the strategy begin by simulating the \textbf{reading and writing by the tape head} from the current state~$q$, represented by~$\loc_q$. This location has an invariant $\clock \leq n$ and an uncontrollable self-loop that resets~$\clocky$ as soon as it reaches~$n$.
The reading of~$m$ and writing of~$m'$ correspond to a transition, leaving~$\loc_q$, with a label among $\zerocell \rightarrow \zerocell$, $\zerocell \rightarrow \onecell$, $\onecell \rightarrow \zerocell$ and $\onecell \rightarrow \onecell$.
This transition resets~$\clock$ and can be taken when $\clock = n$ and $\clocky < n$, so that it is synchronized with the reciting of the content of the $\cur$-th cell of the tape (see \cref{item:tape-recitation}). The reached location then gives the choice of the next state in~$\machine$ among $q'$ and~$q''$.

\item The \textbf{choice between $q'$ and $q''$} in~$\machine$ is modelled by two transitions in~$\TA$, each of them with the respective observable label $q1$ or~$q2$. If $q \in \StatesExists$, then the strategy may choose which state is effectively reached between $q'$ and~$q''$, so both transitions are controllable with respective control labels $q1$ and~$q2$.
Otherwise, $q \in \StatesForall$ and the TA chooses the destination: both transitions are uncontrollable.
At this point, no time has elapsed so $\clock = 0$ and $\clocky$ gives the position of the tape head.
The next $n$ time units are dedicated to its move.

\item Suppose \wlogen{} the chosen state is~$q'$. The \textbf{move of the tape head} is achieved by shifting the value of the clock~$\clock$.
A transition represents the choice to move the tape head to the left (resp.\ the right) with control and observable label~$L$ (resp.~$R$), and guard $x = 0$.
They lead to a location with invariant $\clock < n$ and an uncontrollable self-loop resetting~$y$ when equal to~$n$.
This location may lead to a bad location producing the label \textit{out\_of\_tape} providing that $\clocky = n-1 = \clock$ in the $R$ side, or $\clocky = 1$ and $\clock \neq 0$ in the $L$ side. If $\clocky = n-1 \neq \clock$ in the $R$ side (resp.\ $\clocky = 1$ and $\clock \neq 0$ in the $L$ side), the clock $\clock$ is reset, and the next transition goes toward~$q'$ when $\clocky$ reaches~$n$, resetting~$\clocky$.
In both cases, the new value of $\clock$ is the updated value $n - \cur$.
\item The location~$\loc_{q_f}$, that corresponds to the goal state of $\machine$, ensures the opacity of the traces reaching it by (uncontrollably and unobservably) branching in zero time between two final locations, exactly one of which being private.
\end{itemize}

Crossing this part of the TA takes $2n$ time units. \cref{item:tape-recitation,item:countdown} indeed requiring $n$ time units, they can be done one after the other during a step of this first part of the~TA.
\item \label{item:tape-recitation} To guarantee that the tape is correctly simulated, the strategy is forced to \textbf{repeat the tape's content} every $2n$ time units, one cell per time unit. This means that the set of enabled control label at the $i$-th time unit modulo~$n$ contains the reading and writing symbol (\eg{} $\zerocell \rightarrow \onecell$) of the $i$-th cell in the tape. This update occurs at the same time as in the part simulating the transition. This works in a very similar way to the countdown gadget (see \cref{item:countdown}).
\item \label{item:countdown} A \textbf{countdown} forces the strategy to try reaching~$q_f$, and prevents it to turn indefinitely into the TA. This countdown enumerates the configurations of $\machine$, reaching a doubly exponential number. Assume this number is~$2^{n}$. The idea is to reset non-deterministically a clock~$\clock$ to choose a bit among the $n$~bits encoding the number of configurations. These bits thus correspond to possible runs on the~TA, shifted and spread over an interval of $n$ time units (the $i$-th bit is controlled at the times $i$ modulo~$n$); and the strategy needs to ensure the opacity of each of these runs by correctly incrementing the number they encode. After the reset of~$\clock$, a location denoted by $\loc_{(0)}$ is reached (a bit is equal to zero if its run is there).
A self-loop with control label $[0-0]$ and a transition to a location~$\loc_{(1)}$ encode the update of the bit, with the guard $\clock = n$ and reset set $\set{\clock}$. Similarly, two transitions ($[1-1]$ and $[1-0]$) start from $\loc_{(1)}$. Two other transitions labelled by $[1-0]$ (from $\loc_{(0)}$, one with same guard and reset, the other requiring $\clock = 1$) produce a special symbol~$\extraAction$ representing a carry over.
A third similar transition leads to a private location with guard $\clock = 1$, and the only way for the strategy to opacify the produced trace is to allow one of the two first transitions, enabling either $[0-1]$ or $[1-0]$ at the next time unit (forcing the increment of the next bit). Almost the same mechanism is employed to force the strategy to increment at least a bit per cycle of $2n$ time units, and to produce non-opaque traces when all bits are equal to~1.
\item A simple gadget constraints the strategy to enable at each time at most one reading/writing symbol, one state choice (among $q1$ and~$q2$), one tape head move symbol ($R$ or~$L$), and one symbol of bit update for the countdown. It produces a non-opaque trace as soon as two non-compatible symbols are together enabled as they cross two successive transitions in zero time.

\end{enumerate}
This construction takes polynomial space, as the exponential memory as well as the doubly exponential countdown are fully contained in the strategy. If a strategy ensuring opacity of the built TA exists, then it follows a run in~$\machine$ that reaches a final state and it never produces a non-opaque trace, which corresponds to a failure in the simulation of~$\machine$. Conversely if the final state is reachable in~$\machine$, then the corresponding run gives a strategy successfully simulating the alternating Turing machine, concluding this reduction.
\section{Proof for Observable Sequential Strategies}\label{sec:annexOSS}

\OSSprop*
\begin{proof}[Proof (by contraposition)]
Assume that there exists a trace $w$ such that $w\in \PrivateTr{\TA}{\strategyBlocksAll}$ and $w \notin \PublicTr{\TA}{\strategyBlocksAll}$.
Let $\strategy$ be a sequential strategy. Let us show that $w\in \PrivateTr{\TA}{\strategy}$ and $w \notin \PublicTr{\TA}{\strategy}$.

First, as $\strategyBlocksAll$ is the strategy enabling the smallest number of transitions,
we have that $\CRuns{\TA}{\strategyBlocksAll}\subseteq \CRuns{\TA}{\strategy}$ and thus
$w\in \PrivateTr{\TA}{\strategy}$.

Now let $\rho$ be a $\strategy$-compatible run that is not $\strategyBlocksAll$-compatible (if no such run exists, then we immediately have that $w \notin \PublicTr{\TA}{\strategy}$).
The run $\rho$ hence includes a transition associated to a controllable action~$\caction$.
By the OSS setting, $\caction$ is distinguishable, and thus associated to an action~$\action$ specific to~$\caction$.
As $w$ is the trace of a run that is $\strategyBlocksAll$-compatible, it does not contain~$\caction$.
Thus the trace of~$\rho$ is not~$w$.
As this holds for every $\strategy$-compatible run that is not $\strategyBlocksAll$-compatible, we have that $w \notin \PublicTr{\TA}{\strategy}$.
\end{proof}

\OSS*
\begin{proof}
In~\cite{DQY25}, the authors show how to decide in \EXPSPACE the weak opacity of a TA under buffered observations in the absence of control.
Moreover, from \cref{cor:eptystrat}, to decide the existence of a control strategy making a given TA weakly opaque, it suffices to test whether the strategy~$\strategyBlocksAll$ makes the TA weakly opaque.
We can represent the effect of strategy~$\strategyBlocksAll$ on the TA by removing all controllable actions from the~TA.
Hence, we only need to test the weak opacity of the resulting TA under buffered observations, without control.
This can thus be done in \EXPSPACE.

Regarding hardness, it is known that weak opacity is \EXPSPACE-hard for discrete-time TAs without control~\cite{ADL26}.
As a discrete TA produces by construction buffered observations (at most one observation by time unit), we immediately get that the weak opacity is \EXPSPACE-hard for TAs under buffered observations (with and without control).
\end{proof}

\OSSRC*

In order to establish this result, we show successive restrictions that can be
assumed about the sequential strategy.

Given a run~$\rho$, we say a sequential strategy~$\strategy$ is \emph{$\rho$-minimal} whenever it forbids every controllable action outside the path corresponding to~$\rho$.

Formally, we define the sequence of time intervals $(I_i)_{i=0,\dots,k}$ such that 
(1) for all $n\in \mathbb{N}$ with $2n\leq k$, $I_{2n} =\{n\}$,
(2) for all $n\in \mathbb{N}$ with $2n+1\leq k$, $I_{2n+1} =(n;n+1)$, and
(3) $\rho$ reach the final location during the time interval $I_k$?
Then, we have that, for each buffered observation~$w$ and each time interval~$I$, either there exists $i\leq n$ such that $(I,w)=(I_i, \PTrace{\run}{I_i})$, 
or $\strategy(w)=\emptyset$.
We first show that this shape of sequential strategies is sufficient to decide the existence of \reachcondition{} sequential strategies ensuring opacity.

\begin{lemma}\label{th:simpstrat1}
Given a TA~$\TA$, there exists a \reachcondition{} sequential strategy making $\TA$ weakly opaque iff there
exists a run~$\run$ and a $\run$-minimal \reachcondition{} sequential strategy~$\strategy$ making $\TA$ weakly opaque.
\end{lemma}
\begin{proof}
Let $\TA$ be a TA and $\strategy$ be a \reachcondition{} sequential strategy making $\TA$ weakly opaque.
As $\strategy$ is \reachcondition{}, $\CRuns{\TA}{\strategy}$ is not empty, and thus contains a run~$\run$.
We define $\strategy_\run$ as the $\run$-minimal sequential strategy such that,
decomposing $\run$ into a sequence $(\run_1,I_1),\dots,(\run_n,I_n)$ such that $\run=\run_1\dots \run_n$ and $\run_i$ occurs during the time interval $I_i$, 
we have for all $i\leq n$, $\strategy_\run(I_i, \PTrace{\run}{I_i})=\strategy(I_i, \PTrace{\run}{I_i})$. For every other pair, $\strategy_\run$ is the empty set, as per definition of a $\run$-minimal sequential strategy.
Let us show that $\strategy_\run$ is a \reachcondition{} sequential strategy making $\TA$ weakly opaque.

First, it is \reachcondition{}: as $\strategy_\run$ imitates $\strategy$ over the partial traces of
$\run$, $\run$ is $\strategy_\run$-compatible and thus $\run \in \CRuns{\TA}{\strategy_\run}$.

Now assume there exists a trace $w$ such that $w\in \PrivateTr{\TA}{\strategy_\rho}$ and 
$w \notin \PublicTr{\TA}{\strategy_\rho}$.
As in \cref{cor:eptystrat}, $\strategy_\rho$ being more restrictive than $\strategy$
we have that $w\in \PrivateTr{\TA}{\strategy}$.
Let us show that $w \notin \PublicTr{\TA}{\strategy}$.
We decompose $w=w_1\notick w_2 \tick \dots w_m t$ with $t\in \{\tick,\notick\}$ (in other words, $w_i$ is produced during the time interval $I_i$).
Assume first that $m\leq n$ and for all $i\leq m-1$, $w_i=\PTrace{\run}{I_{i}}$, then 

$\strategy$ and $\strategy_\rho$ behaving identically over this sequence of partial trace,
$w \notin \PublicTr{\TA}{\strategy}$.
Otherwise either $m>n$ or there exists $i\leq m-1$ such that $w_i=\PTrace{\run}{I_{i}}$.
Let $j=n+1$ in the first case and $j=i+1$ in the second.
We have that $\strategy_\run(I_j, w_1\notick w_2 \tick \dots w_{j-1})=\emptyset$, and the same will hold for any future time interval.
Hence, as in \cref{cor:eptystrat}, the controllable actions being distinguishable in the OSS setting and as $w$ is the trace of a run that is
$\strategy_\run$-compatible (due to $w\in \PrivateTr{\TA}{\strategy_\run}$), for all $k\geq j$
$w_k$ does not contain any controllable action. Therefore, any $\strategy$-compatible 

run with trace~$w$ is also $\strategy_\rho$~compatible, and thus
$w \notin \PublicTr{\TA}{\strategy}$, implying that $\strategy$ does not make $\TA$ weakly opaque. This is a contradiction, and thus $\strategy_\rho$ makes $\TA$ weakly opaque.
\end{proof}

We now strengthen this result by limiting the number of events that occur during each time interval within the selected run.

\begin{lemma}\label{th:simpstrat2}

Given a TA~$\TA$ with $h$~clocks, $n$~locations, and where every constant appearing in a constraint is bounded by~$M$,
there exists a run~$\rho$ and a $\rho$-minimal \reachcondition{} sequential strategy
$\strategy$ which makes $\TA$ weakly opaque iff there exists a run
$\rho'$ such that there exists a $\rho'$-minimal \reachcondition{} sequential strategy $\strategy'$ which makes $\TA$ weakly opaque and during each time interval~$I$, $\rho'$ has at most $N=2^{(2M+2)^h h!n}$ controllable events during $I$.
\end{lemma}
\begin{proof}
We will rely here on the definitions from \cref{appendix:proof:NSS}, in particular on the notion of natural belief.

Let $\rho$ be a run and $\strategy$ be a $\rho$-minimal \reachcondition{} sequential strategy which makes $\TA$ weakly opaque.
Due to the OSS setting, the set of regions that can be reached within~$\TA$ after a word~$w$ does not depend on the strategy (the strategy only selects which words can be produced).
Hence, if we denote $w=\tw(\run)$ and assume that $w$ can be decomposed into 
$w_1w_2w_3$ such that $w_2\in \ObsSet^*$ (which ensures that $w_2$ is read within a single time interval), and 
$\beliefcontrol{w_1 w_2}{\strategy}=\beliefcontrol{w_1}{\strategy}$.
We have
$\beliefcontrol{w_1 w_2w_3}{\strategy}=\beliefcontrol{w_1w_3}{\strategy'}$
where $\strategy'$ is obtained by imitating $\strategy$, assuming a $w_2$ is added if the buffered word is prefixed by~$w_1$:
$\strategy'(w_1w') = \strategy'(w_1w_2w')$ and  
$\strategy'(w') = \strategy'(w')$ if $w_1$ is not a prefix of~$w'$.

Let $\strategy'$ be the sequential strategy obtained by repeating the above contraction until
it stabilises (the process ends as $\rho$ is finite and $\strategy$ cannot be simplified 
over other buffered observations as it forbids every controllable action) and let $w'$ the word obtained from~$w$ at the end of the process.
We thus have 
$\beliefcontrol{w}{\strategy}=\beliefcontrol{w'}{\strategy'}$.
As $\run$ reaches a final location, there thus exists a final location in 
$\beliefcontrol{w'}{\strategy'}$, which means that $\strategy'$ is \reachcondition{}.
Moreover, as the natural belief generated through $\strategy'$ are the same as the ones generated by~$\strategy$, it ensures weak opacity of~$\TA$.

We have ensured that there does not exist a decomposition of~$w'$ such that the same
natural belief is encountered twice within a single time interval. 
The number of natural beliefs being bounded by $N=2^{(2M+2)^h h!n}$, the number of events  in~$w'$ (and \emph{a fortiori} the number of controllable events as every controllable event is observable in the OSS setting) occurring within a single time interval is bounded by~$N$.
\end{proof}

We can finally proceed to the proof of \cref{theorem-OSS-NB}.

\begin{proof}[Proof of \cref{theorem-OSS-NB}]
By \cref{th:simpstrat1,th:simpstrat2}, we can limit the search of sequential strategies ensuring weak opacity to sequential strategies that are $\rho$-minimal for some
$\rho$ and such that at most $N=2^{(2M+2)^h h!n}$ controllable events occur in 
a single time interval.
This second point implies that there exists a strategy ensuring weak opacity iff there 
exists an $N$-sequential strategy ensuring weak opacity.
We can thus rely on \cref{theorem:NSS} to decide the existence of this strategy.

The algorithm from \cref{theorem:NSS} is in \TwoEXPTIME. %
Hence, as~$N$, being given in binary, is simply exponential,
this reduction to the \ProblemNSS{} problem produces a \ThreeEXPTIME algorithm.

Regarding the hardness, since the proof of \cref{subsection:NSShard} falls into the OSS setting, it applies immediately.
\end{proof}

\end{document}